\documentclass[twoside]{article}

\usepackage{aliascnt}
\usepackage{amsthm}
\usepackage[colorlinks = true, citecolor = blue]{hyperref}
\usepackage{graphicx} 
\usepackage{xcolor}
\usepackage{amsmath, amssymb}
\usepackage[algo2e, inoutnumbered, algoruled, vlined]{algorithm2e}
\usepackage[utf8]{inputenc}
\usepackage{cleveref}
\makeatletter

\crefname{theorem}{theorem}{Theorems}
\Crefname{Theorem}{Theorem}{Theorems}
\newtheorem*{lemma_nonumber*}{Lemma}

\newtheorem{lemma}{Lemma}

\newaliascnt{corollary}{theorem}
\newtheorem{corollary}[corollary]{Corollary}
\aliascntresetthe{corollary}
\crefname{corollary}{corollary}{corollaries}
\Crefname{Corollary}{Corollary}{Corollaries}

\newaliascnt{proposition}{theorem}

\aliascntresetthe{proposition}
\crefname{proposition}{proposition}{propositions}
\Crefname{Proposition}{Proposition}{Propositions}

\newaliascnt{definition}{theorem}

\aliascntresetthe{definition}
\crefname{definition}{definition}{definitions}
\Crefname{Definition}{Definition}{Definitions}

\newaliascnt{remark}{theorem}

\aliascntresetthe{remark}
\crefname{remark}{remark}{remarks}
\Crefname{Remark}{Remark}{Remarks}

\crefname{example}{example}{examples}
\Crefname{Example}{Example}{Examples}

\crefname{figure}{figure}{figures}
\Crefname{Figure}{Figure}{Figures}

\newtheorem{assumption}{\textbf{H}\hspace{-3.5pt}}
\crefformat{assumption}{{\textbf{H}}#2#1#3}

\usepackage{bbold}
\usepackage{kky}
\usepackage{geometry}
\usepackage{algorithm}
\usepackage{subcaption}
\usepackage{comment}
\usepackage{thm-restate}
\usepackage{multirow}
\usepackage{tikz}

\usepackage[round]{natbib}

\usepackage{kky}

\usepackage[accepted]{aistats2025}

\begin{document}

\newcommand{\hr}[1]{{\color{red} #1}}
\runningtitle{Distribution-Aware Mean Estimation under User-level LDP}

\twocolumn[

\aistatstitle{Distribution-Aware Mean Estimation \\under User-level Local Differential Privacy}

\aistatsauthor{ Corentin Pla \And Hugo Richard \And  Maxime Vono }

\aistatsaddress{ Criteo AI Lab \And  Criteo AI Lab \And Criteo AI Lab }]

\begin{abstract}
We consider the problem of mean estimation under user-level local differential privacy, where $n$ users are contributing through their local pool of data samples.
Previous work assume that the number of data samples is the same across users.
In contrast, we consider a more general and realistic scenario where each user $u \in [n]$ owns $m_u$ data samples drawn from some generative distribution $\mu$; $m_u$ being unknown to the statistician but drawn from a known distribution $M$ over $\mathbb{N}^\star$.
Based on a distribution-aware mean estimation algorithm, we establish an $M$-dependent upper bounds on the worst-case risk over $\mu$ for the task of mean estimation. We then derive a lower bound. The two bounds are asymptotically matching up to logarithmic factors and reduce to known bounds when $m_u = m$ for any user $u$.
\end{abstract}

\section{Introduction}

In the current artificial intelligence era, machine learning (ML) techniques and algorithms have become common tools for data analysis and decision making.
In most cases, their empirical success is based on the availability of large training datasets, which is for instance theoretically justified by data probabilistic modelling and the Bernstein-von Mises approximation asserting that under appropriate conditions, the more data can be processed, the more accurate the inference can be performed \citep{Bardenet17,Bottou_SIREV_2018,pmlr-v97-cornish19a}.
However, in the meantime, growing user concerns and regulatory frameworks have emerged regarding how user personal data could be collected and processed by third-party entities.
These data privacy issues hence have incentivised users of machine learning applications to ask for methods preserving data anonymity \citep{Shokri15,Abadi16,zaeem}.

Whilst many privacy-preserving tools have been proposed \citep{pmlr-v151-garcelon22a,NEURIPS2021_27545182,10.5555/3241094.3241143,pmlr-v54-mcmahan17a,VonoQLSD}, the \emph{workhorse} framework in machine learning is differential privacy (DP).
DP allows to formally protect user's privacy in a quantifiable manner while maintaining the statistical utility. 
It dates back at least to~\citet{warner1965randomized} but has received a renewed interest after the work of~\citet{dwork_algorithmic_2013}.
As an example, DP has been implemented in real-world scenarii by large technology companies such as Meta, Apple, or LinkedIn \citep{10.1145/2660267.2660348,journals/corr/abs-1709-02753,DVN/TDOAPG_2020,article_linkedin}; public bodies such as the United States Census Bureau \citep{10.1145/3219819.3226070}; or non-governmental organisations such as Wikimedia Foundation \citep{2023arXiv230816298A}.
The DP framework, which requires outputs of a \emph{differentially-private} algorithm to be undistinguishable when a single contribution changes, has been formalised first in the so-called central model where data is aggregated before being privatised \citep{Dwork,10.1145/1806689.1806786}.
The main ingredient to ensure such guarantees is to randomise seminal algorithms by incorporating calibrated noise.
To embrace more stringent privacy guarantees, DP has then evolved into a distributed setting, called \emph{local} differential privacy (LDP), which privatises data before aggregation, alleviating the need to resort to a trusted aggregator \citep{10.1007/978-3-642-31594-7_39,doi:10.1080/01621459.2017.1389735}. 

In contrast to the central setting, LDP comes with a degraded privacy/utility trade-off mainly due to noise involved in each local data contribution.
Such fundamental limits have been established by a series of recent works, focusing on specific tasks and associated algorithms.
For instance, the problem of finding optimal private protocols has already been addressed for mean estimation \citep{duchi2018minimax}, density estimation \citep{butucea2019localdifferentialprivacyelbow}, and hypothesis testing \citep{berrett2020locallyprivatenonasymptotictesting}.

Traditionally, differentially-private models based on LDP assume that users are contributing with only one data sample \citep{duchi2018minimax}. 
Whilst this framework is sufficient for basic use-cases such as mean wage estimation of several individuals, it cannot be applied to a broader set of real-world applications where users may contribute via multiple data points.
As seminal examples, users can contribute with a local database made of multiple item (\emph{e.g.}, product or movie) ratings in order to train recommendation systems \citep{10.1145/1557019.1557090}; or can hold several sequences of words typed on their mobile device keyboards which are then used to train next-word prediction models \citep{McMahan2017LearningDP}.  
Such general applications nourished a research area focusing on \emph{user-level} differential privacy, where each user holds more than one data sample and seeks to maintain the privacy of her entire collection.
User-level DP has been considered under both central and local settings for statistical tasks already pointed out previously, that is mean estimation, empirical risk minimisation, and non-parametric density estimation \citep{Girgis,acharya2022discretedistributionestimationuserlevel,kent2024rateoptimalityphasetransition}.
Unfortunately, the user-level DP paradigm considered in aforementioned state-of-the-art research work lacks realism as it involves an unrealistic assumption of uniformity over the number of samples.
More precisely, current work assume that the number of data samples provided by each user is the same, which does not hold in many real-world settings as outlined,  for instance, by federated learning applications \citep{mcmahan17,FLreview2021}. 

This paper aims at filling this gap.
Among statistical tasks that have been tackled in the literature, \emph{univariate mean estimation} stands for the primary and most interesting problem to consider as it is the basis to many other protocols. 
Indeed, associated results can be extended to multivariate mean estimation via specific data transformations involving Hadamard matrices \citep{NEURIPS2020_222afbe0}, non-parametric density estimation via the trigonometric basis \citep{butucea2019localdifferentialprivacyelbow}, sparse mean estimation \citep{acharya2022discretedistributionestimationuserlevel} or regression via stochastic convex optimisation \citep{pmlr-v202-bassily23a}. 

As such, we consider in this paper the problem of univariate mean estimation under user-level local differential privacy where each user $u \in [n]$ locally holds a dataset of $m_u \in \mathbb{N}^\star$ independent and identically distributed observations from an unknown probability distribution $\mu$; $m_u$ being drawn from a known discrete probability distribution $M$.
To the best of our knowledge, we are the first to investigate this general user-level local DP setting.

\noindent \textbf{Contribution.} Our contribution is three-folded.
\begin{itemize}
    \item We propose a novel and more realistic user-level LDP framework that allows users to contribute via local data sets of heterogeneous cardinality.
    
    \item We instantiate this framework for the task of univariate mean estimation and propose an associated algorithm, coined \emph{distribution-aware mean estimation} (\texttt{DAME}). 
    The proposed algorithm works in two phases: (i) the \emph{localisation} phase aims at finding a bin where the true mean parameter lies with high probability; (ii) the \emph{estimation} phase projects data samples onto the bin to reduce the range. 
    Compared to other variants proposed in previous work, the novelty of \texttt{DAME} lies in (i) the selection of users and the bin size, which are adapted to the data set size distribution $M$; (ii) the estimation phase where the mean is shrinked towards the center of the bin to include data contributions from all users even those having few samples.
    
    \item We derive non-asymptotic guarantees on the worst-case risk over $\mu$ for mean estimation.
    More precisely, we derive upper bounds via the proposed algorithm \texttt{DAME}, and also provide lower bounds on the same quantity. 
    Based on the derived bounds, we show that \texttt{DAME} is optimal in many scenarii: (i) asymptotically in $n$ up to a log-factor; (ii) in the item-level ($m_u=1$ for any $u \in [n]$) LDP setting, up to constant factor; (iii) in the homogeneous ($m_u=m \in \mathbb{N}^\star$ for any $u \in [n]$) user-level LDP setting, up to a log factor; and (iv) in many regimes interpolating the item-level and user-level settings. 
    The optimality of the derived bounds is verified numerically for several choices of the discrete distribution $M$.
\end{itemize}

\noindent \textbf{Related Work.} As outlined previously, private mean estimation under the DP paradigm has been investigated in both the central \citep{Dwork,10.1145/1806689.1806786} and local settings \citep{duchi2018minimax}.
These seminal research works have focused on providing item-level (in contrast to user-level) privacy under the assumption that the number of samples is the same across all users.
When the data samples are discrete, the more general problem of frequency estimation has also been well-studied under the same assumptions in both central \citep{Dwork} and local DP models \citep{10.1007/978-3-642-31594-7_39,10.1145/2660267.2660348,pmlr-v89-acharya19a}.
User heterogeneity in terms of number of data samples per user has already been considered in both different DP settings and statistical tasks.
For instance, \citet{NEURIPS2020_f06edc8a,levy2021learning}
studied the problem of learning discrete distributions with different
number of samples per user.
However, the aforementioned works used a sub-optimal approach to tackle this user heterogeneity by discarding samples of some users and using the median number of other users to apply known estimation techniques derived in the homogeneous user setting.

Regarding private mean estimation, at least two recent work considered different number of samples per user but under central DP.
More precisely, \citet{10.5555/3600270.3602383} took inspiration from private federated learning settings and proposed a private mean estimation algorithm where users first compute local mean estimates which are then aggregated and privatised by a central aggregator.
In contrast to our proposed approach, their methodology is non-interactive and very different in terms of algorithmic steps -- leading to difficulties to translate their central DP work to our local DP setting. 
In addition, \citet{NEURIPS2024_huber} recently proposed a novel private mean estimation approach, based on Huber loss minimisation, which aims at tackling the potential bias of the two-phase procedure we consider in our algorithm \texttt{DAME}.

\section{Problem Formulation}
\label{sec:Problem Formulation}
\noindent \textbf{Notation.} We set $[n] = \{1, \ldots, n\}, [n_1, n_2] = \{n_1, \ldots, n_2\}$ and use the notation $z_{i:j} = (z_i, \cdots, z_j)^\top$ and $Z^{(i:j)} = (Z^{(i)}, \cdots, Z^{(j)})^\top$. We define $\Bcal([-1, 1])$ as the Borel set over $[-1, 1]$ and let $\mu^{\otimes i}$ denote the $i$-th product distribution.
Through this paper, we will use the mathematical notation, $\Ocal(\cdot)$, $\Omega(\cdot)$ and $\Theta (\cdot)$ to describe the complexities of the computation required by private mean estimation algorithms.
We recall that $f(d) = \Ocal(g(d))$ (resp. $f(d) = \Omega(g(d))$)  if there exists $c > 0$ such that $f(d) \leq cg(d)$ (resp. $f(d) \geq cg(d)$). 
We use $f(d) = \Theta(g(d))$ if there exist $c_1, c_2 > 0$ such that $c_1 g(d) \leq f(d) \leq c_2g(d)$.
We define $a \wedge b = \min(a,b)$ and $a \vee b = \max(a,b)$. Finally we denote $D_{TV}(\cdot, \cdot)$ the total variance distance and $D_{KL}(\cdot, \cdot)$ the Kullback-Leibler divergence. 

\noindent \textbf{Framework.} We study the setting of user-level local differential privacy, in which users do not trust the central data aggregator.
We consider an environment involving $n \in \mathbb{N}^\star$ users holding \emph{private} local data sets $\{X^{(u)} = \{X^{(u)}_1,\ldots,X^{(u)}_{m_u}\}\}_{u \in [n]}$ of different sizes $\{m_u\}_{u \in [n]}$.
For any $u \in [n]$, $m_u$ is distributed according to a discrete distribution $M$ defined on $(\NN^\star, 2^{\NN^\star})$, and $X^{(u)}$ stands for a random variable distributed according to a probability distribution $\mu$ defined on a measurable space $(\Xcal, \mathfrak{X})$.
For the sake of simplicity and similar to previous work, we assume, for any $u \in [n]$ and $t \in [m_u]$, that $X_t^{(u)} \in [-1,1]$.
Note that in constrast to \emph{item-level} LDP where $m_u=1$ for any $u \in [n]$, the considered \emph{user-level} LDP setting assumes that each user $u \in [n]$ contributes $m_u \in \NN^*$ samples to the global (artificial) dataset so that $\Xcal = \bigcup_{m=1}^{\infty} [-1, 1]^m$.
As a result, each user local data set is sampled independently and identically according to the joint distribution $\nu_\mu$ described by the following generative model:
\begin{align}
    \label{eq:nu_2}
    m_u \sim M, ~ X^{(u)} \ | \ m_u=i \sim \mu^{\otimes i} \ .
\end{align}

In order to preserve their private local datasets, users do not directly reveal $X^{(u)}$ to the statistician but instead send to the latter an obfuscated view $Z^{(u)}$ of $X^{(u)}$, defined over $(\Zcal, \mathfrak{Z})$.
More precisely, the random variables $\{Z^{(u)}\}_{u \in [n]}$ stand for differentially-private versions of $\{X^{(u)}\}_{u \in [n]}$. 
We assume that users sequentially reveal $\{Z^{(u)}\}_{u \in [n]}$ to the statistician. 
This assumption means that for any user $u \in [n]$, $Z^{(u)}$ and $\{X^{(u')}\}_{u' \neq u}$ are independent given $X^{(u)}$ and $Z^{(1:u-1)}$. 
Such independence property implies that the conditional probability distribution $Q$ of $Z^{(1:n)}$ given $X^{(1:n)}$, referred to as \emph{mechanism}, is fully characterised by $\{Q_u\}_{u \in [n]}$ where, for any $u \in [n]$, $Q_u$ is the local conditional probability distribution of $Z^{(u)}$ given $X^{(u)}, Z^{(1:u-1)}$, coined \emph{local channel}.

\noindent \textbf{Local Differential Privacy.} For a given scalar privacy parameter $\alpha > 0$ and any $u \in [n]$, we say that the random variable $Z^{(u)}$ is an \textit{$\alpha$-differentially locally private} view of $X^{(u)}$ if, for any $z_{1:u-1} \in \Zcal^{u-1}$ and $x, x^{\prime} \in \Xcal$, the conditional probability distribution $Q_u: \mathfrak{Z} \times (\Xcal \times \Zcal^{u-1}) \rightarrow [0, 1]$ of $Z^{(u)}$ given $X^{(u)}, Z^{(1:u-1)}$ satisfies:

\begin{equation}
\sup _{S \in \mathfrak{Z}} \frac{Q_u\left(S \mid x, z_{1:u-1} \right)}{Q_u\left(S \mid x^{\prime}, z_{1:u-1}\right)} \leq \exp (\alpha) \ .
\end{equation}

Note that the previous definition does not constrain $Z^{(u)}$ to be a data release exclusively based on $X^{(u)}$: the local channel $Q_u$ may be interactive, that is differing based on previous data releases $Z^{({1:u-1})}$. 
When for any $u \in [n]$, $Z_u$ is an $\alpha$-differentially locally private view of $X_u$, we say that $Q = \{Q_u\}_{u\in[n]}$ is an \emph{$\alpha$-LDP mechanism}. We refer to $\Qcal_\alpha$ as the set of $\alpha$-LDP mechanisms.

\noindent \textbf{Private Mean Estimation.} The objective of mean estimation under user-level LDP is to estimate $\theta = \EE_{X \sim \mu}[X]$ by choosing (i) an $\alpha$-LDP mechanism $Q$ that will generate data $Z^{(1)}, \dots, Z^{(n)}$ from $X^{(1)}, \dots, X^{(n)}$, and (ii) an estimator $\hat \theta: \Zcal^n \rightarrow [0, 1]$. 
The difficulty of such estimation task is measured by a metric referred to as the \emph{worst-case risk over $\mu$}, and defined by:
\begin{align}
    \label{eq:risk}
    R_{\alpha, n, M} = \inf _{Q \in \mathcal{Q}_\alpha, \hat \theta \in \Theta} \sup_{\mu \in \Dcal} \EE\left[\left | \hat{\theta}\left(Z^{({1:n})}\right) -\theta \right |^2\right].
\end{align}

This paper presents in \Cref{sec:algo} a practical algorithm, referred to as \texttt{DAME}, to perform private mean estimation.
In \Cref{sec:theory}, we derive upper and lower bounds on the worst-case risk \eqref{eq:risk} -- the upper bound being reached by the proposed algorithm. 

\section{Distribution-Aware Mean Estimation (\texttt{DAME})}
\label{sec:algo}

In this section, we present the proposed algorithm to perform private mean estimation under user-level LDP.
We refer to the latter as distribution-aware mean estimation (\texttt{DAME}) algorithm, to emphasise that it works under user data set size heterogeneity.

\noindent \textbf{High-level Description.} Similar to other interactive algorithms proposed in previous work \citep{berrett2020locallyprivatenonasymptotictesting,acharya2022discretedistributionestimationuserlevel,acharya2022role,kent2024rateoptimalityphasetransition,berrett2020locallyprivatenonasymptotictesting}, \texttt{DAME} proceeds in a two-phase procedure.
In the first phase, coined the \emph{localisation phase}, private data of the first half of users is discretised and privatised via a calibrated randomised response mechanism \citep{warner1965randomized}. 
This phase aims to identify a candidate bin where the true mean lies with high probability. 
In the second phase, coined the \emph{estimation phase}, private data of remaining users is projected onto the candidate bin (slightly enlarged) and privatised via the Laplace mechanism \citep[Definition 3.3]{dwork_algorithmic_2013} where the added noise scales with the width of the enlarged candidate bin.
The novelty of \texttt{DAME} is that it adapts the bin size to the data set size distribution $M$, and uses a novel biasing/debiasing procedure in the estimation phase to include data from all users even those that have few samples. 
The following paragraphs provide additional details regarding the two phases involved in \texttt{DAME}.

\noindent \textbf{Algorithmic Details.} More precisely, we assume an even number of users $n$, discarding one user if we need to, and use the first $n/2$ users for the localisation phase and the last $n/2$ for the estimation phase. 
The interval $[-1,1]$ is partitioned into non-overlapping sub-intervals $I_j$ (also called \emph{bins}) of width $2\tau$, where $\tau > 0$ is given by:
\begin{equation}
    \tau = \sqrt{\frac{2 \log(8 (\sqrt{\tilde{m} n \alpha^2} \vee 1))}{\tilde{m}}} \ ,
    \label{eq:tau}
\end{equation}
with $\tilde{m} \in \NN^\star$ being the \emph{effective maximum data set size} and standing as an hyper-parameter of the \texttt{DAME}. 
We denote by $l$ the bin index containing the true mean $\theta$, that is $\theta \in I_l$.

\textcircled{\raisebox{-0.9pt}{1}} In the localisation phase, each user $u \in [n/2]$ such that her number of samples $m_u$ verifies $m_u \geq \tilde{m}$ computes her sample mean $\bar{X}^{(u)}_{m_u} = (1/m_u)\sum_{t=1}^{m_u} X^{(u)}_t$, identifies in which sub-intervals $\bar{X}^{(u)}_{m_u}$ falls into, and computes the indicator vector $V^{(u)} = (V_j^{(u)})_{j \in [\lceil \frac{1}{\tau}\rceil]}$ defined, for any $j \in [\lceil \frac{1}{\tau} \rceil]$, by
\begin{equation}
    V^{(u)}_j = \ind{\bar{X}^{(u)}_{m_u} \in \cup_{k \in \{j-1, j, j+1\}} I_k} \ .
    \label{eq:Vuj}
\end{equation} 
Informally, we say that user $u$ \emph{votes} for the $j$-th bin if $V^{(u)}_j = 1$. 
A user therefore votes for the bins that either contain its empirical mean or are neighbors of the bin containing its empirical mean. 
This ensures that $I_l$, the interval containing the true mean $\theta$, receives a vote with high probability even when $\theta$ is close or equal to the border of the bin. 
The indicator vector of users $u \in [n/2]$ with $m_u \leq \tilde{m}$ is given by $V^{(u)} = 0$. In other words, users with low number of samples are not allowed to take part to the vote.

Then, a privatised version $\tilde{V}^{(u)}$ of $V^{(u)}$ is computed using the randomised response mechanism for binary vectors with three non-zero entries~\cite[Lemma 15]{kent2024rateoptimalityphasetransition}.  
More precisely, for any $j \in [\lceil \frac{1}{\tau} \rceil]$, $\tilde{V}_j^{(u)}$ is defined as:
\begin{equation}
    \tilde{V}_j^{(u)} = \begin{cases}
    V_j^{(u)} &\text{with proba.} \ \pi_{\alpha/6} = \frac{e^{\alpha / 6}}{1+e^{\alpha / 6}} \ , \\
    1-V_j^{(u)} &\text{otherwise.} 
    \end{cases} 
\label{eq:privatization-sheme}
\end{equation}
The candidate bin $I_{\hat j}$ is the one that receives the highest number of votes, that is:
\begin{equation}
    \hat{j}  = \underset{j \in [\lceil \frac{1}{\tau} \rceil]}{\argmax } \sum_{u = 1}^{n/2} \widetilde{V}_j^{(u)} \ .
\end{equation}

\textcircled{\raisebox{-0.9pt}{2}} In the estimation phase, each remaining user $u \in [n/2 +1, n]$ computes an estimate of her respective local empirical mean that is shrunk towards the mid-point $s_{\hat j}$ of $I_{\hat j}$, and defined as:
\begin{equation*}
    \label{eq:hat x}
    \hat{X}^{(u)}_{\hat j} =  \frac{\sqrt{m_u \wedge \tilde{m}}}{\sqrt{\tilde{m}}} \left(\bar{X}^{(u)}_{m_u} + \left(\frac{\sqrt{\tilde{m}}}{\sqrt{m_u \wedge \tilde{m}}} - 1\right) s_{\hat{j}}\right). 
\end{equation*}

Then, $\hat{X}^{(u)}_{\hat j}$ is projected onto $[L_{\hat j}, U_{\hat j}]$, an enlarged copy of $I_{\hat j}$, where $6 \tau$ are added to the left and right borders to ensure that shrunk estimates are not affected by the projection, with high probability. 
A Laplace noise scaling with the width of $[L_{\hat j}, U_{\hat j}]$ is then added to preserve privacy.
For any $u \in [n/2+1,n]$, the resulting local mean estimate for user $u$ writes:
\begin{equation}
    \hat{\theta}^{(u)} = \Pi_{\hat{j}}\left(\hat{X}^{(u)}_{\hat j}\right) + \frac{14 \tau}{\alpha} \ell_u \ ,
    \label{eq:return-user}
\end{equation}
where $\Pi_{\hat{j}} \left(.\right)$ is the projection onto $[L_{\hat j}, U_{\hat j}]$.

Lastly, the estimate of the true mean $\theta$ is given by:
\begin{equation}
    \hat{\theta}= \frac{\bar{\theta} \sqrt{\tilde{m}} -  \sum_{i=1}^{\tilde{m}} \left(\sqrt{\tilde{m}} - \sqrt{i}\right) M (i) s_{\hat{j}}}{ \mathbb{E}_{m \sim M}\left(\sqrt{m \wedge \tilde{m}}\right)} \ ,
\label{eq:hat theta}
\end{equation}
where $\bar{\theta} = \frac{2}{n} \sum_{u=\frac{n}{2}+1}^n \theta^{(u)}$ and $M(i) = \PP_{m \sim M}(m=i)$.

\noindent \textbf{Algorithm.} The main steps of the proposed algorithm \texttt{DAME} are depicted in \Cref{algo:dame}.

\begin{algorithm}[!ht]
\caption{Distribution-Aware Mean Estimation (\texttt{DAME}). The algorithm is described from a centralised point of view. However, the statistician never has access to user data, unless explicitly shared.}
\SetKwComment{Comment}{\color{green!50!black} $\triangleright$\ }{}

\vspace{0.2cm}

\SetKwInput{KwData}{Input}
\KwData{number of users ($n$), privacy parameter ($\alpha$), effective maximum data set size ($\tilde{m}$) and distribution of data set size ($M$).}

\vspace{0.2cm}

Compute $\tau$ according to \eqref{eq:tau}.

\vspace{0.1cm}

Divide $[-1, 1]$ into non-overlapping intervals $\{I_j = [l_j, u_j)\}_{j \in [\lceil \frac{1}{\tau} \rceil]}$ of size $2 \tau$.

\vspace{0.2cm}

\Comment{\color{green!50!black} Localisation phase: users with large data sets elect privately as candidate the bin most likely to contain $\theta$.}
\vspace{0.1cm}

\For{$u \in [n/2]$}
{
\If{$m_u \geq \tilde{m}$}
{
Compute $\bar{X}^{(u)}_{m_u} = \frac1{m_u} \sum_{t=1}^{m_u} X_t^{(u)}$ \ .

Compute $V^{(u)}_j = \ind{\bar{X}^{(u)}_{m_u} \in \cup_{k \in \{j-1, j, j+1\}} I_k}$ for any $j \in [\lceil \frac1{\tau} \rceil]$ \ .
}
\Else
{
Set $V^{(u)} = 0$
}
Share $\tilde{V}_j^{(u)} = \begin{cases}
    V_j^{(u)} &\text{w.p. } \frac{e^{\alpha / 6}}{1+e^{\alpha / 6}} \\
    1-V_j^{(u)} &\text{otherwise} 
    \end{cases} $; $j \in [\lceil \frac1{\tau} \rceil]$.
}

\vspace{0.1cm}

Compute $\hat{j}  = \underset{j \in [\lceil \frac{1}{\tau} \rceil]}{\argmax } \sum_{u = 1}^{n/2} \widetilde{V}_j^{(u)}$ \ .

Compute $L_{\hat j} = (l_j - 6 \tau) \vee -1, U_{\hat j} = (u_j + 6 \tau) \wedge 1$ and $s_{\hat j} = \frac{l_j + u_j}{2}$ \ .

\vspace{0.1cm}

\Comment{\color{green!50!black} Estimation phase: empirical mean of users with small data sets are biased toward the elected bin.}
\vspace{0.1cm}

\For{$u \in [\frac{n}{2}+1, n]$}
{
Set $\hat{X}^{(u)}_{\hat j} = \frac{\sqrt{m_u \wedge \tilde{m}}}{\sqrt{\tilde{m}}} \left(\bar{X}^{(u)}_{m_u} + \left(\frac{\sqrt{\tilde{m}}}{\sqrt{m_u \wedge \tilde{m}}} - 1\right) s_{\hat{j}}\right)$ \ .
\vspace{0.1cm}
Sample $l_u$ from standard Laplace distribution.
\vspace{0.1cm}
Share $\hat{\theta}^{(u)} = \Pi_{\hat{j}}\left(\hat{X}^{(u)}_{\hat j}\right) + \frac{14 \tau}{\alpha} \ell_u$ \ .
}
\vspace{0.1cm}
Compute $\bar{\theta} = \frac{2}{n} \sum_{u=\frac{n}{2}+1}^n \theta^{(u)}$ \ .

\vspace{0.1cm}
\Comment{\color{green!50!black} Bias correction.}
\vspace{0.1cm}
\Return $\hat{\theta} = \frac{\bar{\theta} \sqrt{\tilde{m}} -  \sum_{i=1}^{\tilde{m}} \left(\sqrt{\tilde{m}} - \sqrt{i}\right) M (i) s_{\hat{j}}}{ \mathbb{E}_{m \sim M}\left(\sqrt{m \wedge \tilde{m}}\right)}$ \ .
	\label{algo:dame}
\end{algorithm}

\section{Theoretical Analysis}
\label{sec:theory}

In this section, we present an upper bound, obtained via \texttt{DAME} detailed in \Cref{algo:dame}, on the worst-case risk $R_{\alpha, n, M}$ defined in \eqref{eq:risk}. 
Furthermore, we also derive a lower bound on $R_{\alpha, n, M}$. 
These two non-asymptotic bounds allow us to discuss the asymptotic optimality of our results, while drawing connections with known results from the private mean estimation literature.

\subsection{Lower and Upper Bounds on $R_{\alpha, n, M}$}

\noindent \textbf{Assumption.} In order to derive our non-asymptotic theoretical results, we consider the following assumptions on the main parameters characterising the considered distribution-aware user-level LDP paradigm.

\begin{assumption}[Finite expectation $M$]
\label{ass:distribution}
The data set size distribution $M$ is known and such that $\EE_{m \sim M}[m] < \infty$.
\end{assumption}

\begin{assumption}[Bounded support of $\mu$]
\label{ass:bounded support}
The common data distribution $\mu$ admits a known support defined as $[-1, 1]$.
\end{assumption}

\begin{assumption}[High privacy regime for $\alpha$]
\label{ass:high privacy}
The LDP parameter $\alpha > 0$ is set such that $\alpha \leq 22/35$.
\end{assumption}

\Cref{ass:distribution} is a rather weak assumption satisfied by most discrete distributions encountered in practice. \Cref{ass:bounded support} is a common assumption in local differential privacy made for instance in~\cite{duchi2018minimax, bassily2019linear, blasiok2019towards}. In the worst-case, the square dependency on the range cannot be avoided but for some distributions, better results can be obtained~\citep{bun2019average}. Lastly, the high privacy regime in \Cref{ass:high privacy} is a rather strong assumption that we make for the sake of simplicity.
This assumption mainly allows to re-use  results derived in \cite{duchi2018minimax}, that also rely on \Cref{ass:high privacy}. 
We believe our results could be extended to lower privacy regimes but leave this task to future work.

\noindent \textbf{Lower bound.} \Cref{th:lb2} below provides a lower bound on the worst-case risk $R_{\alpha,n,M}$ defined in \eqref{eq:risk}.

\begin{restatable}[Lower bound]{theorem}{lbtwo}
\label{th:lb2}

Assume \Cref{ass:distribution}-\ref{ass:high privacy}.
Then, there exist $c_1, c_2 > 0$, independent of $\alpha$, $n$ and $m$, such that the following lower bound holds:
\begin{align}
    R_{\alpha, n, M} \geq \max_{a \in \NN}  \frac{c_1 e^{-c_2 n \alpha^2 \PP_{m \sim M}(m >a)^2}}{n \alpha^2 \EE_{m \sim M}[\sqrt{m}\ind{m \leq a}]^2 \vee 1} \ .
    \label{eq:lower_bound}
\end{align}
The positive constants $c_1$ and $c_2$ are explicitly given in \Cref{app:lb2}.
\end{restatable}
\begin{proof}[Proof sketch, see \Cref{app:lb2} for more details]
    Consider probability distributions $\mu_0$ and $\mu_1$ supported on $\{-1, 1\}$ such that $\mu_0(1) = (1 - \delta/2)$ and $\mu_1(1) = (1 + \delta/2)$. Consider for $i \in \{0, 1\}$, $\nu_i = \nu_{\mu_i}$ where $\nu_{\mu}$ is described in \cref{eq:nu_2}. From a sequential application of LeCam's bound, Bretagnolle-Huber (\cite{Yu1997}) and \cite[Th 1]{duchi2014localprivacydataprocessing}, it holds that $R_{n, \alpha, M} \geq \frac{\delta^2}{4} \exp(-12 n \alpha^2 D_{TV}(\nu_0, \nu_1)^2)$. Then Lemma \ref{lemma:TV} gives $D_{TV}(\nu_0, \nu_1) \leq \EE_{m \sim M}[\sqrt{\frac{m}{2} D_{KL}(\mu_0, \mu_1)} \wedge 1]$ and Lemma \ref{lemma:kl-tv} gives $D_{KL}(\mu_0, \mu_1) \leq 3 \delta^2$ for $\delta \in [0, \frac12]$. Then choosing $\delta^2 = \frac1{4 (n\alpha^2 \EE_{m \sim M}[\sqrt{m} \ind{m \leq a}]^2 \vee 1)}$ and optimising over $a$ gives the lower bound.
\end{proof}

The lower bound in \eqref{eq:lower_bound} means that there exists a universal constant $c_2$ such that given $M, \alpha, n$, and for any user-level LDP algorithm returning an estimate $\hat \theta$ of $\theta = \EE_{X \sim \mu}[X]$, there exists a choice of $\mu$ such that for any $a \in \NN$, it holds that
\begin{align*}
    \EE[|\hat \theta - \theta |^2] = \Omega\left(\frac{\exp(-c_2 n \alpha^2 \PP_{m \sim M}(m >a)^2)}{n \alpha^2 \EE_{m \sim M}[\sqrt{m}\ind{m \leq a}]^2 \vee 1}\right) \ .
\end{align*}

\noindent \textbf{Upper bound.} Our next result, provided in \Cref{th:ub2}, analyses the performance of \texttt{DAME} and therefore provides an upper bound on the worst-case risk $R_{n, \alpha, M}$ defined in \eqref{eq:risk}.

\begin{restatable}[Upper bound]{theorem}{ubtwo}
\label{th:ub2}
Assume \Cref{ass:distribution}-\ref{ass:high privacy}.
Let $c_3, c_4, c_5 > 0$ be universal constants, independent of $n$, $\alpha$ and $M$.
Consider the function $\phi: a \in \NN^* \mapsto c_5/(n \alpha^2) \log[c_4 (a n\alpha^2 \vee 1)/\log(c_4 (a n\alpha^2 \vee 1))]$ and set $\tilde{m}= \argmax_{a \in \NN^\star}\{\mathbb{P}_{m \sim M}\left( m \geq a \right)^2 \geq \phi(a) \wedge 1\}$. 
Then, we have the following lower bound.
\begin{align}
    R_{\alpha, n, M} \leq  \frac{c_3\ln(c_4 ( \sqrt{\tilde{m} n\alpha^2} \vee 1))}{n \alpha^2 (\mathbb{E}_{m \sim M}(\sqrt{m \wedge \tilde{m}}))^2} \wedge 4 \ .
\end{align}
The positive constants $c_3$, $c_4$ and $c_5$ are explicitly given in \Cref{app:ub2}.
\end{restatable}
\begin{proof}[Proof sketch, see \Cref{app:ub2} for more details.]
The upper bound is provided by \texttt{DAME} which is instantiated with $\tilde{m}$ as specified in \Cref{th:ub2}. The proof starts by a bias-variance decomposition of the objective and we first focus on the bias. Thanks to Hoeffding bounds \citep{409cf137-dbb5-3eb1-8cfe-0743c3dc925f}, each user can estimate $\theta$ up to a precision scaling in approximately $1/\sqrt{m_u}$. Since the bins size scales with $1/\sqrt{\tilde{m}}$, users with data set size higher than $\tilde{m}$ will vote for $I_l$ or its neighboring bin with high probability.  However, because of the privatisation, for each user (including those with low data set size), a bin could gain or lose a vote with probability $1 - \pi_{\alpha/6}$. \Cref{lemma:p(abar)} in the Appendix precisely upper bounds the probability that the localisation is not successful, \emph{i.e.}, $\hat j$ is at distance at least 3 from the optimal bin $l$. 
If the localisation is not successful, the error is upper bounded by $4$. 
If the localisation is successful, a user $u$ with a high number of samples ($m_u \geq \tilde{m}$) is likely to have its empirical mean $\bar{X}^{(u)}_{m_u}$  in $[L_{\hat j}, U_{\hat j}]$ since this interval contains $\theta$ and is of size $1/\sqrt{\tilde{m}}$ up to log factors (see \Cref{lemma:(I)}). Users with low number of samples are unlikely to be in this interval and therefore their estimate $\hat{X}^{(u)}_{m_u}$ of the mean is a linear combination of their empirical mean and $s_{\hat j}$ the middle point of bin $I_{\hat j}$ where coefficients are chosen to guarantee that $\hat{X}^{(u)}_{m_u}$ lies in $[L_{\hat j}, U_{\hat j}]$. 
This procedure adds a bias that can be computed explicitly and is removed when $\theta$ is estimated (see \cref{eq:hat theta}) therefore ensuring that the error added by users with low number of samples is controlled (see \Cref{lemma:iii} and \Cref{lemma:ii}). 
We then focus on the variance (see \Cref{lemma:(V)}) which is given by $\tilde{m}\Var(\hat \theta^{(u)})/\EE_{m \sim M}[\sqrt{m \wedge \tilde{m}}]$ by definition of $\hat \theta$. 
The variance of $\hat \theta^{(u)}$ is the sum of the variance of the Laplace noise and the variance of the estimates, the latter being controlled by the projection. 
The precise formula for $\tilde{m}$ trades-off the probability of finding the correct bin in the localisation phase and the error terms due to the variance.
\end{proof}

To reach the upper bound presented in \Cref{th:ub2}, $\tilde{m}$ must be set as the solution of the optimisation problem
\begin{equation}
    \argmax_{a \in \NN^*}\{\underbrace{\mathbb{P}_{m \sim M}\left( m \geq a \right)^2 - \phi(a) \wedge 1}_{\psi(a)} \geq 0\} \ . \label{eq:tildem}
\end{equation}
Notice that $\psi$ is decreasing, $\psi(1) \geq 0$ and $\psi(\exp(n\alpha^2)) \leq 0$. It follows that $\tilde{m}$ can be found via a binary search with less than $2 \lceil n \alpha^2 \rceil$ iterations.

The optimality of our bounds and comparison with previous work is discussed in \Cref{subsec:discussion}.

\subsection{Discussion}
\label{subsec:discussion}

In this section, we demonstrate that our bounds are asymptotically optimal up to log-factors.
In addition, we show that they reduce to known bounds in the item-level LDP setting or the homogeneous (\emph{i.e.}, $m_u=m \in \NN^\star$) user-level LDP setting. 
Lastly, we illustrate numerically our theoretical insights by plotting upper and lower bounds for classical distributions and studying the empirical performance of \texttt{DAME} on synthetic data.

\noindent \textbf{Asymptotic Optimality.} First, we claim that \texttt{DAME} is asymptotically optimal up to a log factor as the number of users $n$ grows towards infinity.
Indeed, by taking $a \rightarrow \infty$ in \Cref{th:lb2}, we get the following lower bound:
\begin{align}
    R_{\alpha, n, M} = \Omega\left(\frac1{n \alpha^2 \EE_{m \sim M}[\sqrt{m}]^2}\right).
    \label{eq:lb-item-level}
\end{align}
Looking at the upper bound in \Cref{th:ub2}, since, for any $a \in \NN$, $\lim_{n \rightarrow \infty} \phi(a) = 0$, we get that $\tilde{m} \rightarrow \infty$. 
Furthermore, we have $\lim_{n \rightarrow \infty} \tilde{m}/(n \alpha^2) \leq 1$. Indeed, $\tilde{m} \geq n \alpha^2$ implies that, for any $x \leq n \alpha^2$, $\PP(m \geq x) \geq \sqrt{\phi(x)} \geq \sqrt{\frac1{n \alpha^2}}$ leading to $\EE_{m \sim M}[m] \geq \int_{0}^{n \alpha^2} \PP_{m \sim M}(m \geq x) dx \geq \sqrt{n \alpha^2}$, which would diverge and therefore violate \Cref{ass:distribution}. 
Hence, it follows that there exists a universal constant $c > 0$, such that:
\begin{align}
    \lim_{n \rightarrow \infty} R_{\alpha, n, M}\left(\frac{\log(n \alpha^2)}{n \alpha^2 \EE_{m \sim M}[\sqrt{m}]^2}\right)^{-1} \leq c \ ,
    \label{eq:ub-item-level}
\end{align}
where $R_{\alpha, n, M}$ is defined in \eqref{eq:risk}.

Combining \eqref{eq:lb-item-level} and \eqref{eq:ub-item-level}, it follows that the proposed private mean estimation algorithm \texttt{DAME} asymptotically attains tight rates up to log factors. 
We now show that \texttt{DAME} is tight in settings studied in previous work.

\noindent \textbf{Comparison with Previous Work - General Setting.}
\begin{table}
\begin{tabular}{lcc}
\hline Setting & Minimax rate & Reference  \\
\hline
 $m=1$ & {\small$\left(n  \alpha^2\right)^{-1}$} & {\small\cite{duchi2018minimax}} \\
 $m \in \mathbb{N}^\star$  & {\small$\left(m n  \alpha^2\right)^{-1}$}  & {\small \cite{kent2024rateoptimalityphasetransition}} \\
 $m \sim M $ & {\small$\left(\mathbb{E}_{m \sim M}[\sqrt{m})]^2n \alpha^2\right)^{-1}$} & {\small\texttt{DAME} (this work)}\\ 
\hline
\end{tabular}
\caption{General setting -- asymptotic minimax rate with respect to $\mu$ (as $n \rightarrow +\infty$). Results are given discarding log factors.}
\label{table:comp}
\vspace{-0.1cm}
\end{table}
In the item-level LDP setting, the data set size distribution $M$ is given by $M(1) = 1$, where we recall that, for any $i\in \NN^\star$, $M(i) = \PP_{m \sim M}(m=i)$. 
Optimal rates are given by $\Theta(1/n\alpha^2)$, see \citet[Corollary 1]{duchi2018minimax}. 
This lower bound coincides with that of \Cref{th:lb2} as shown in \eqref{eq:lb-item-level}.
A naive application of our upper bound in \Cref{th:ub2} uses $\tilde{m} = 1$ and gives a rate of oder $\mathcal{O}(\log(n\alpha^2)/(n\alpha^2))$ which is only tight up to a log factor. 
However, as argued in the proof of \Cref{th:ub2}, when $\tilde{m}=1$, there is only one bin equal to the full interval $[-1, 1]$. 
In this case, \texttt{DAME} becomes identical to the optimal procedure derived in~\citet{duchi2018minimax} and therefore reaches the same upper bound. 

\Cref{table:comp} summarises our findings regarding the comparison of our bounds derived in \Cref{th:lb2} and \Cref{th:ub2} with previous work, in terms of asymptotic optimality.

\noindent \textbf{Comparison with Previous Work - Toy Example.} We now consider a more involved setting where for $\rho \in [0, 1]$, $M(m) = \rho, M(1) = (1 - \rho)$. 
In this scenario, taking $a=0$ and $a=m$ in the lower bound defined in \Cref{th:lb2} gives:
\begin{align}
    &R_{\alpha, n, M} = \nonumber\\
    &\Omega\left(\frac1{n \alpha^2 ((1 - \rho) + \rho \sqrt{m})^2} \vee \frac{e^{-c_2 n \alpha^2 \rho^2}}{(1 - \rho)^2 n \alpha^2 \vee 1}\right).
    \label{eq:2spikes}
\end{align}
The associated upper bound is given by \Cref{cor:toy}, detailed below.
\begin{corollary}[Upper bound]
\label{cor:toy}
Assume \Cref{ass:distribution}-\ref{ass:high privacy}.
For $\rho \in [0, 1]$, take $M$ such that $M(m) = \rho, M(1) = (1 - \rho)$ and assume $n \alpha^2 \geq 1$. Then, $\phi(a) = c_5 (n \alpha^2)^{-1} \log[c_4 (a n\alpha^2)/\log(c_4 (a n\alpha^2))]$ and we have the following upper bounds:
\begin{align*}
    &\forall \rho^2 \in [0, 1], R_{n \alpha, M} = \mathcal{O}\left(\frac{\ln(n \alpha^2)}{n \alpha^2}\right) \\
    &\forall \rho^2 \geq \phi(m), R_{n, \alpha, M} = \mathcal{O}\left(\frac{\log(m n \alpha^2)}{n \alpha^2 ((1 - \rho) + \rho \sqrt{m})^2}\right) \\
    &\forall a \in \{2, \dots, m-1\}, \forall \rho^2 \in [\phi(a), \phi(a+1)),\\
    & ~~~~~~ R_{n, \alpha, M} = \mathcal{O}\left(\frac{\log(a n \alpha^2)}{n \alpha^2 ((1 - \rho) + \rho \sqrt{a})^2}) \wedge \frac{e^{-\rho^2 \frac{ n \alpha^2}{c_5}}}{\rho^2}\right)
\end{align*}
\end{corollary}

\begin{proof}
    For any $\rho$, we have $\tilde{m} \geq 1$ which gives the upper bound in $\log(n \alpha^2)/(n \alpha^2)$. 
    When $\rho^2 \geq \phi(m)$, then $\tilde{m} = m$ and the upper bounds follows. 
    Otherwise, the condition on $\rho^2$ ensures that $\tilde{m}=a$ which gives the upper bound in $\log(a n \alpha^2)/(\rho^2 a n \alpha^2)$. 
    Then the condition $\rho^2 \leq \phi(a + 1)$ implies $\rho^2 \leq \phi(2a)$ which leads to $\log(2c_4 a n \alpha^2)/(2 c_4 a n \alpha^2) \leq \exp(-\rho^2 n \alpha^2/c_5)$.
    This implies the bound in $\exp(-\rho^2 n \alpha^2/c_5)/\rho^2$.
\end{proof}
\Cref{cor:toy} assumes $n \alpha^2 \geq 1$. This is not restrictive as when $n \alpha^2 \leq 1$ the trivial upper bound $4$ matches up to a constant factor the lower bound in $\Omega(1/(n\alpha^2))$. 
When $\rho \geq \phi(m)$, the upper and lower bounds derived in \eqref{eq:2spikes} and \Cref{cor:toy} are tight up to a log factor. In particular, when $\rho=1$ and $\phi(m) \leq 1$, we get an upper bound in $\log(m n \alpha^2)/(m n \alpha^2)$ and a matching lower bound up to a log factor. 
On the other hand, when $\rho=1$ and $\phi(m) \geq 1$, the upper bound becomes $\exp(-  n \alpha^2/c_5)$ whereas the lower bound scales in $\exp(c_2 n \alpha^2)$.
Note that $\rho=1$ corresponds to the homogeneous user-level setting ($M(m) = 1$). In this setting, the current known upper and lower bounds~\cite[Theorem 6]{kent2024rateoptimalityphasetransition} are respectively given by $\Omega(\exp(-c'_1 n \alpha^2) \wedge (m n \alpha^2)^{-1})$ and $\Ocal(\exp(-c'_2 n \alpha^2) + \log(m n \alpha^2)/(m n \alpha^2))$, where $c'_1, c'_2$ are universal constants. 
These known results are again matching our bounds.

\Cref{tab:toy} compares the upper bound of \texttt{DAME} to the upper bounds obtained in~\cite{kent2024rateoptimalityphasetransition}.

\begin{table}
\begin{tabular}{cc}
\hline  Regime & \texttt{DAME}  \\
\hline
 $0 \leq \rho^2 < \phi(2)$                    & $\left(n  \alpha^2\right)^{-1}$  \\
 $\phi(2) \leq \rho^2 < \phi(3)$              & $\left(\rho \sqrt{2} + (1 - \rho))^2 n  \alpha^2\right)^{-1}$ \\
 $\vdots$                                    & $\vdots$ \\
 $\phi(a) \leq \rho^2 < 1 $                   & $\left((\rho \sqrt{a} + (1 - \rho))^2 n  \alpha^2\right)^{-1}$  \\
 $1 \leq \rho^2 < \phi(a+1)$                 & $\left(n  \alpha^2\right)^{-1}$ \\ 
\hline
\end{tabular}
\caption{Toy example -- upper bound on worst-case risk with respect to $\mu$ when $M(1) = (1 - \rho)$ and $M(m) = \rho$. The guarantees are given discarding log factors.
Known bounds in the literature are $(n \alpha^2)^{-1}$ in all regimes.}
\label{tab:toy}
\vspace{-0.3cm}
\end{table}

\begin{figure*}[h!]
\begin{subfigure}{0.45\textwidth}
\centering
\includegraphics[width=\linewidth]{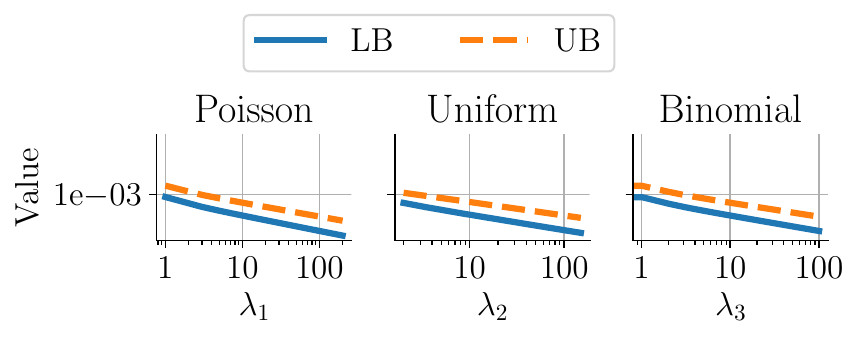}
\includegraphics[width=\linewidth]{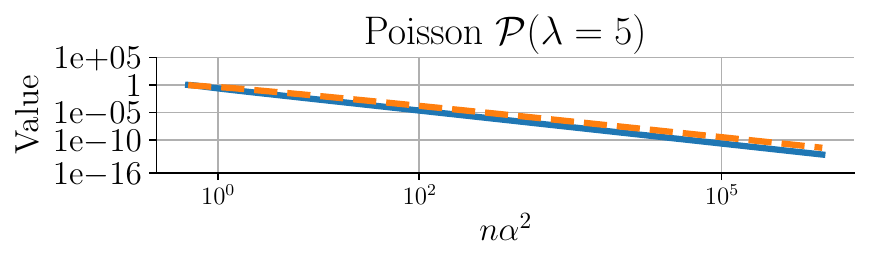}
\caption{\emph{Top:} UB \& LB for varying parameters of Poisson, Uniform, and Binomial distributions. \emph{Bottom:} UB \& LB for varying $n \alpha^2$}
\label{fig:subim1}
\end{subfigure}\quad \quad \quad 
\begin{subfigure}{0.45\textwidth}
\centering
\includegraphics[width=\linewidth]{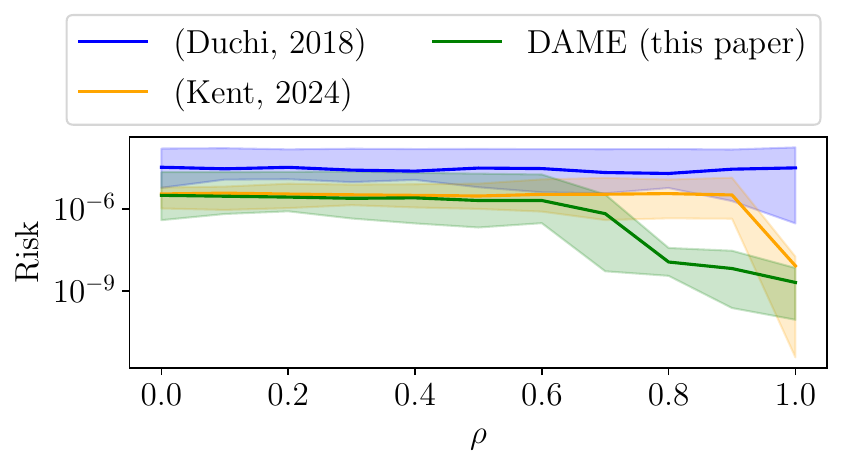}
\caption{ Risk of the estimator given by~\cite{duchi2018minimax}, \cite{kent2024rateoptimalityphasetransition} and \texttt{DAME} respectively when $M(m_1) = 1-\rho$ and $M(m_2) = \rho$.}
\label{fig:subim2}
\end{subfigure}
\label{fig:distributions}
\end{figure*}

\noindent \textbf{Empirical study.} We plot in \Cref{fig:subim1} the upper and lower bounds for various distributions which allows us to verify the tightness of our bounds and then benchmark \texttt{DAME} with \cite{duchi2018minimax} and \cite{kent2024rateoptimalityphasetransition} on a synthetic dataset.

More precisely, \Cref{fig:subim1} (top)  shows the upper bound (dotted orange) and lower bound (blue) when $M$ is chosen according to either Poisson $\Pcal(\lambda_1)$, Uniform $\Ucal({1, 2\lambda_2- 1})$ or Binomial $\Bcal(1000,\lambda_3/1000)$. We take $n\alpha^2 = 500$ and vary $\lambda_1, \lambda_2, \lambda_3$ in a grid of $100$ points. 
We note that the upper bounds and lower bounds form two parallel straight lines in a log-log plot, attesting to the log-tight nature of our bounds. 
To draw these curves, we need to estimate $a_1$ resp. $a_2$ such that $a_1 = \argmax_{a \in \NN^\star}\{ \phi(a)\}$ and $a_2$ maximises the lower bound. 
The parameter $a_1$ is found via a binary search and $a_2$ via a grid search.

We then fix the distribution $M$ to a Poisson with parameter $\lambda=5$ and vary $n\alpha^2$ in a grid of $100$ points. The results are displayed in \Cref{fig:subim1} (bottom) and show that the upper and lower bounds behave similarly with respect to $n \alpha^2$.

Lastly, we benchmark \texttt{DAME}, \cite{duchi2018minimax} and \cite{kent2024rateoptimalityphasetransition} on a synthetic dataset where $\mu$ is such that $\mu(1) = \mu(-1) = 0.5$ and the data set size distribution $M$ is defined by $M(m_1) = 1 - \rho$ and $M(m_2) = \rho$ where $m_1 = 10^5, m_2 = 10^6$ and $\rho$ varies on a linear grid with $10$ points between $0$ and $1$. 
We use $n=10^4, \alpha=22/35$ and $\tilde{m}$ is chosen via a binary search. 
The results, displayed in \Cref{fig:subim2}, shows that \texttt{DAME} is able to obtain better results than other approaches because it adapts to the distribution, benefiting from the presence of users with many samples even when some users have few of samples.

\section{Conclusion}
In this paper, we consider the mean estimation problem under user-level LDP constraints in the case where each user $u \in [n]$ owns $m_u$ data samples drawn from some generative distribution $\mu$; $m_u$ being unknown to the statistician but drawn from a known distribution $M$ over $\mathbb{N}$. 
We propose a mean estimation algorithm, coined \texttt{DAME}, and derive non-asymptotic guarantees on the worst-case risk over $\mu$ presented in \Cref{th:lb2} and \Cref{th:ub2}. 
These results show that \texttt{DAME} provides asymptotic optimality up to some log-factors and matches previously known bounds~\citep{duchi2018minimax, kent2024rateoptimalityphasetransition}.

As univariate mean estimation is a sub-routine used in the current best-known algorithms for multivariate mean estimation, sparse estimation, or distribution estimation~\citep{duchi2014localprivacydataprocessing, acharya2022role} in item-level LDP, we believe \texttt{DAME} can easily be used as a building block to provide efficient algorithms for various tasks in user-level LDP without the need to assume that all users have an identical dataset size -- an assumption not met in practice.

An important assumption of this work is the perfect knowledge of $M$ (\Cref{ass:distribution}) that is used for the debiasing step of \Cref{algo:dame} and in the choice of $\tilde{m}$ defined in \eqref{eq:tildem}, leading to the upper bound in \Cref{th:ub2}. 
Whether inexact knowledge of $M$ can be used to achieve similar bounds is a question of interest.

Lastly, \texttt{DAME} is interactive. Whether there exists a non-interactive algorithm with similar good properties could be the subject of future work.

\bibliography{bibliography/biblio}
\bibliographystyle{plainnat}

\appendix
\onecolumn

\theoremstyle{plain}
\newtheorem{unlemma}{Lemma S}
\newtheorem{unproposition}{Proposition S}
\newtheorem{uncorollary}{Corollary S}
\newtheorem{untheorem}{Theorem S}
\setcounter{equation}{0}
\setcounter{figure}{0}
\setcounter{table}{0}
\setcounter{page}{1}
\makeatletter
\renewcommand{\theequation}{S\arabic{equation}}
\renewcommand{\thefigure}{S\arabic{figure}}
\renewcommand{\thetable}{S\arabic{table}}
\renewcommand{\thetheorem}{S\arabic{theorem}}
\renewcommand{\thelemma}{S\arabic{lemma}}
\renewcommand{\thesection}{S\arabic{section}}
\renewcommand{\theremark}{S\arabic{remark}}
\renewcommand{\theproposition}{S\arabic{proposition}}
\renewcommand{\thecorollary}{S\arabic{corollary}}

\section{Proof of \Cref{th:lb2}} 
\label{app:lb2}

\lbtwo*

The explicit values of the constants are $c_1 = \frac{e^{-9}}{16}$ and $c_2 = 24$.

We start by the following result which relates the total variation of $\nu_0$ and $\nu_1$
to the total variation between $\mu_0$ and $\mu_1$.

\begin{lemma}[Decomposition of TV distance by number of users]
\label{lemma:TV}
Consider $\nu_0 = \nu_{\mu_0}$ and $\nu_1 = \nu_{\mu_1}$ where $\nu_\mu$ is defined as in Equation~\eqref{eq:nu_2}. The following property holds:
    \begin{align}
        D_{TV}(\nu_0, \nu_1)^2 \leq \left( \sum_{m=1}^{\infty} M(m) \min(\sqrt{\frac{m}{2}D_{KL}(\mu_0, \mu_1)}, 1) \right)^2
    \end{align}
\end{lemma}
\begin{proof}
For the proof of Lemma~\ref{lemma:TV}, the following result will be used. 
\begin{lemma}[TV and KL distance between $\mu_0$ and $\mu_1$]
\label{lemma:kl-tv}
  Define $\mu_0$ and $\mu_1$ to be probability measures on $\{-1, 1\}$ where
  $\mu_0(\{1\}) = (1- \delta)/2 $, and $\mu_0(\{1\}) = (1 + \delta)/2$. It holds that
  \begin{align}
    &D_{TV}(\mu_0, \mu_1) = \delta  \label{eq:TV-delta}\\
    &D_{KL}(\mu_0, \mu_1) \leq 3\delta^2 \text{ for } \delta \in [0,1/2] \label{eq:KL-delta}
  \end{align} 
\end{lemma}
\begin{proof}[Proof of Lemma~\ref{lemma:kl-tv}]
  Let us start by showing \eqref{eq:TV-delta}. 
  \begin{align}
      D_{TV}(\mu_0, \mu_1) &= \frac{1}{2}\left( \mid \mu_0(\{1\}) - \mu_1(\{1\})\mid+ \mid \mu_0(\{0\}) - \mu_1(\{0\})\mid \right) \\&
      = \frac{1}{2}\left(\mid (1- \delta)/2 -(1 + \delta)/2  \mid + 
 \mid (1 + \delta)/2 - (1- \delta)/2\mid\right) \\&=\delta 
  \end{align}

  Now let us prove \eqref{eq:KL-delta} 

  \begin{align}
      D_{KL}(\mu_0, \mu_1) &= \mu_0(\{1\}) \cdot \ln\left( \frac{\mu_0(\{1\})}{\mu_1(\{1\})}\right) +\mu_0(\{0\}) \cdot \ln\left( \frac{\mu_0(\{0\})}{\mu_1(\{0\})}\right) \\&
      = \delta \ln\left( \frac{1+\delta}{1-\delta}\right)\\&
      \leq 3 \delta^2 \text{ for } \delta \in [0,1/2]
  \end{align}
where the last inequality comes come since the function $f(\delta) = \ln(\frac{1 + \delta}{1 - \delta}) - 3 \delta$ is negative on $[0, \frac12]$.
Indeed, $f(0) = 0$ and 
\[
f'(\delta) = \underbrace{\frac1{1 + \delta}}_{\leq 1} + \underbrace{\frac1{1 - \delta}}_{\leq 2} - 3 \leq 0.
\]
\end{proof}

First let us show that 

    \[
    D_{TV}(\nu_0, \nu_1) = \sum_{m=1}^{\infty} M(m) D_{TV}(\mu_0^i, \mu_1^i)
    \]
    
\begin{align}
    D_{TV}(\nu_0, \nu_1) &= 2 \sup_{A \in \mathcal{X}}\left|\nu_0(A) - \nu_1(A)\right| \\&= 2 \sup_{A \in \mathcal{X}} \left| \nu_0\left(\bigcup_{i=1}^{+\infty } A_i\right) - \nu_1\left(\bigcup_{i=1}^{+\infty } A_i\right)  \right| ~~~ \text{(Since $A = \bigcup_{i=1}^{+ \infty}A_i,\ A_i \in [-1,1]^i \ \forall i \in \mathbb{N}^* $)}
    \\& = 2 \sup_{A \in \mathcal{X}} \left| \sum_{i=1}^{+\infty}\left(\nu_0\left( A_i \cap \{m=i\} \right) - \nu_1\left( A_i\cap \{m=i\}\right) \right) \right|   \\& = 2 \sup_{A \in \mathcal{X}}\left|  \sum_{i=1}^{+\infty}p_i\left(\mu_0^i\left(A_i\right)  - \mu_1^i\left(A_i\right)\right) \right| \\& \leq \sum_{i=1}^{+ \infty}p_i \left( 2\sup_{B \in [-1,1]^i}\left(\mu_0^i\left(B\right)  - \mu_1^i\left(B\right)\right)\right) ~~~ \left(\text{ By triangular inequality }\right) \\& \leq \sum_{i=1}^{+ \infty}p_i D_{TV}(\mu_0^i, \mu_1^i)
    \\& 
\end{align}

    Then, notice that
    \begin{align}
    D_{TV}(\nu_0, \nu_1)^2 & =  \left( \sum_{m=1}^{\infty} M(m) D_{TV}(\mu_0^m, \mu_1^m)\right)^2 
    \\& \overset{(1)}{\leq} \left( \sum_{m=1}^{\infty} M(m) \min(\sqrt{\frac{1}{2}D_{KL}(\mu_0^m, \mu_1^m)}, 1) \right)^2 
    \\& \overset{(2)}{=}\left( \sum_{m=1}^{\infty} M(m) \min(\sqrt{\frac{m}{2}D_{KL}(\mu_0, \mu_1)}, 1) \right)^2
\end{align}
where inequality $(1)$ is by Pinsker's inequality \hr{add ref} and $D_{TV}(\mu_0^m, \mu_1^m) \leq 1$ and equality $(2)$ is by the KL of product distributions.
\end{proof}

We now give the proof of \Cref{th:lb2}.
Define $\mu_0$ and $\mu_1$ to be probability measures on $\{-1, 1\}$ where $\mu_0(\{1\}) = (1 - \delta)/2 $, and $\mu_1(\{1\}) = (1 +  \delta)/2$ where the parameter $\delta \in (0, \frac12)$ will be chosen later. Define $\nu_0 = \nu_{\mu_0}$ and $\nu_1 = \nu_{\mu_1}$ where $\nu_\mu$ is defined as in \Cref{eq:nu_2} and their privatized counterpart $M_0^n, M_1^n$ defined for $i \in \{0, 1\}$ as
\[
M_i^n(S) = \int_{x \in \Xcal} Q(S | x) d\nu_i(x).
\]

Starting from Le Cam's bound \hr{add ref}, we obtain the following results.
\begin{align}
R_{n, \alpha, M} &\geq \frac{\delta^2}{2}\left\{1-D_{\mathrm{TV}}\left(M_0^n, M_1^n\right)\right\} && \text{ (Le Cam's bound)} \\
&\geq \frac{\delta^2}{4} e^{-D_{\mathrm{KL}}\left(M_0^n, M_1^n\right)} && \text{(Bretagnolle-Huber, \hr{add ref})}\\
& \geq \frac{\delta^2}{4} e^{-12 n \alpha^2 D_{\mathrm{TV}}\left(\nu_0, \nu_1 \right)^2} && \text{(\cite[Th 1]{duchi2014localprivacydataprocessing})}.
\end{align}

We then upper bound $D_{\mathrm{TV}}\left(\nu_0, \nu_1 \right)$ starting from \Cref{lemma:TV}.
\begin{align}
    D_{\mathrm{TV}}\left(\nu_0, \nu_1 \right)^2 &\leq \left( \sum_{m=1}^{+\infty} M(m) \min(\sqrt{\frac{m}{2}D_{KL}(\mu_0, \mu_1)}, 1) \right)^2 && \text{(By \Cref{lemma:TV})} \\
    &\leq \left( \sum_{m=1}^{+\infty} M(m) \min(\sqrt{\frac{m}{2} 3 \delta^2}, 1) \right)^2. && \text{(By \Cref{lemma:kl-tv})}
\end{align}

Then for any $a \in \NN^*$ it holds that 
\begin{align}
    D_{\mathrm{TV}}\left(\nu_0, \nu_1 \right)^2 &\leq \left( \sum_{m=1}^{a} M(m) \min(\sqrt{\frac{m}{2} 3 \delta^2}, 1) + \sum_{m=a + 1}^{+\infty} M(m) \min(\sqrt{\frac{m}{2} 3 \delta^2}, 1) \right)^2 \\
    &\leq \left( \sum_{m=1}^{a} M(m) \min(\sqrt{\frac{m}{2} 3 \delta^2}, 1) + \sum_{m=a + 1}^{+\infty} M(m) \right)^2 \\
    &\leq \left( \sqrt{\frac{3 \delta^2}{2}} \EE_{m \sim M}[\sqrt{m}\ind{m \leq a}] + \PP_{m \sim M}( m > a) \right)^2 \\
    &\leq 3 \delta^2 \EE_{m \sim M}[\sqrt{m}\ind{m \leq a}]^2 + 2\PP_{m \sim M}( m > a)^2 \\
\end{align}

Then choosing $\delta^2 = \frac1{ 4 \max( n \alpha^2 \EE_{m \sim M}[\sqrt{m}\ind{m \leq a}]^2, 1)} \leq \frac{1}{4}$  we obtain the bound
\begin{align}
    &\forall a \in \NN^* \\
    &R_{n, \alpha, M} \geq \frac1{16 \max(\EE_{m \sim M}[\sqrt{m}\ind{m \leq a}]^2 n \alpha^2, 1) } \exp( - 9) \exp(-24 n \alpha^2 \PP_{m \sim M}(m >a)^2),
\end{align}
which in turn gives the bound stated in the theorem.

\section{Proof of \Cref{th:ub2}}
\label{app:ub2}

\ubtwo*

\begin{proof}
We will show that $c_3 = 1570$, $c_4 = 8$, $c_5 = 868.5$.
The fact that $\EE[|\hat \theta - \theta|^2] \leq 4$ comes from $\hat \theta, \theta \in [-1, 1]$. 
To get the other term in the bound, we start by writing the following bias-variance decomposition.
\begin{align}
  \EE\left[|\hat{\theta}-\theta|^2\right] &= \EE\left[\EE\left\{|\hat{\theta}-\theta|^2 \mid \hat{j}\right\}\right]
  \\&= \EE\left[\left\{\EE\left(\hat{\theta} \mid \hat{j} \right)-\theta\right\}^2 + \Var\left(\hat{\theta} \mid \hat{j}\right)\right]. \label{eq:bias-var}
\end{align}
where $\hat{\theta}$ is defined Equation \eqref{eq:hat theta}. Let $m$ be a random variable sampled from $M$. It holds that 
\begin{equation}
\hat{\theta}=  \frac{\sqrt{\tilde{m}} \frac{2}{n} \sum_{u=n/2+1}^n\hat{\theta}^{(u)}}{ \mathbb{E}(\sqrt{m \wedge \tilde{m}})} - \frac{ \mathbb{E}((\sqrt{\tilde{m}}-\sqrt{m})\mathbb{1}(m\leq \tilde{m}))}{\mathbb{E}(\sqrt{m \wedge \tilde{m}})}s_{\hat{j}},
\end{equation}

First, let us focus on the biased term. We can write that: 
\begin{align}
\EE\left(\hat{\theta} \mid \hat{j} \right) &=  \EE\left(\frac{\sqrt{\tilde{m}} \frac{2}{n} \sum_{u=n/2+1}^n\hat{\theta}^{(u)}}{ \mathbb{E}_{ }(\sqrt{m \wedge \tilde{m}})} - \frac{ \mathbb{E}_{ }((\sqrt{\tilde{m}}-\sqrt{m})\mathbb{1}(m\leq \tilde{m}))}{\mathbb{E}_{ }(\sqrt{m \wedge \tilde{m}})}s_{\hat{j}} \mid \hat{j}\right)\\& = \frac{\sqrt{\tilde{m}}}{ \mathbb{E}_{ }(\sqrt{m \wedge \tilde{m}})} \EE\left(\frac{2}{n} \sum_{u=n/2+1}^n\hat{\theta}^{(u)} \mid \hat{j} \right) - \frac{ \mathbb{E}_{ }((\sqrt{\tilde{m}}-\sqrt{m})\mathbb{1}(m\leq \tilde{m}))}{\mathbb{E}_{ }(\sqrt{m \wedge \tilde{m}})}s_{\hat{j}} \\& = \frac{\sqrt{\tilde{m}}}{ \mathbb{E}_{ }(\sqrt{m \wedge \tilde{m}})} \EE\left( \hat{\theta}^{(n)} \mid \hat{j} \right)- \frac{ \mathbb{E}_{ }((\sqrt{\tilde{m}}-\sqrt{m})\mathbb{1}(m\leq \tilde{m}))}{\mathbb{E}_{ }(\sqrt{m \wedge \tilde{m}})}s_{\hat{j}} \\&
    \overset{(1)}{=} \frac{\sqrt{\tilde{m}}}{ \mathbb{E}_{ }(\sqrt{m \wedge \tilde{m}})} \EE\left(\Pi_{\hat{j}}\left(\hat{X}^{(n)}_{\hat j} \right) \mid \hat{j} \right) - \frac{ \mathbb{E}_{ }((\sqrt{\tilde{m}}-\sqrt{m})\mathbb{1}(m\leq \tilde{m}))}{\mathbb{E}_{ }(\sqrt{m \wedge \tilde{m}})}s_{\hat j} \notag \\&
    \overset{(2)}{=} \frac{\sqrt{\tilde{m}}}{ \mathbb{E}_{ }(\sqrt{m \wedge \tilde{m}})} \EE\left(\Pi_{\hat{j}}\left(\overline{X}^{(n)}_{m_n}\right)\mathbb{1}(m_n \geq \tilde{m}) + \sum_{i=1}^{\tilde{m}-1} \Pi_{\hat{j}}\left(\overline{X}^{(n)}_{m_n}\right) \mathbb{1}(m_n =i)\mid \hat{j}\right) \notag  \\& - \frac{ \mathbb{E}_{ }((\sqrt{\tilde{m}}-\sqrt{m})\mathbb{1}(m\leq \tilde{m}))}{\mathbb{E}_{ }(\sqrt{m \wedge \tilde{m}})}s_{\hat j} \notag 
\end{align}
where equality $(1)$ comes from the definition of $\hat \theta^{(n)}$ in \eqref{eq:return-user} and equality $(2)$ comes from the definition of $\hat X^{(n)}_{\hat j}$ in \eqref{eq:hat x}.

Then, we can write $\theta = \EE_{X \sim \mu}[X]$ in function of $\overline{X}^{(n)}_{m_n}\ind{m_n \geq \tilde{m}}$ and $(\overline{X}^{(n)}_i)_{i=1}^{\tilde{m}}$.
\begin{lemma}
    Let $\theta = \EE_{X \sim \mu}[X]$, $m \sim M$ . 
    It holds that
    \begin{equation}
        \theta = \frac{\sqrt{\tilde{m}}}{\EE(\sqrt{m \wedge \tilde{m}})}  \EE(\overline{X}^{(n)}_{m_n} \mathbb{1}(m_n\geq \tilde{m})) + \frac{1}{\EE(\sqrt{m \wedge \tilde{m}})}\sum_{i=1}^{\tilde{m}-1}\EE(\overline{X}^{(n)}_{i}\sqrt{i}\mathbb{1}(m_n =i )).
    \end{equation}
    \label{lemma:theta}
\end{lemma}
\begin{proof}
We can write
    \begin{align}
        &\underbrace{\frac{\sqrt{\tilde{m}}}{\EE(\sqrt{m \wedge \tilde{m}})}  \EE(\overline{X}^{(n)}_{m_n}\mathbb{1}(m_n\geq \tilde{m})) + \frac{1}{\EE(\sqrt{m \wedge \tilde{m}})}\sum_{i=1}^{\tilde{m}-1}\EE(\overline{X}^{(n)}_{m_n}\sqrt{i}\mathbb{1}(m_n =i ))}_{(i)} \\& 
        = \sum_{i=\tilde{m}}^{\infty} \frac{\sqrt{\tilde{m}}}{\EE(\sqrt{m \wedge \tilde{m}})}  \EE\left(\overline{X}^{(n)}_{m_n}\ind{m_n = i} \right) + \frac{1}{\EE(\sqrt{m \wedge \tilde{m}})} \sum_{i=1}^{\tilde{m}-1} \sqrt{i} \EE\left(\overline{X}^{(n)}_{m_n} \mathbb{1}(m_n =i )\right).
    \end{align}
    Then notice that 
    \begin{align}
    \EE\left(\overline{X}^{(n)}_{m_n}\ind{m_n = i} \right) &= \EE[\underbrace{\EE[\overline{X}^{(n)}_{i} | m_n = i]}_{=\theta} \ind{m_n=i}] \\
    &= \theta M(i),
    \end{align}
    so that 
    \begin{align}
    (i) &= \frac{\sqrt{\tilde{m}}}{\EE(\sqrt{m \wedge \tilde{m}})}  \sum_{i=\tilde{m}}^{+\infty} \theta M(i) + \frac{1}{\EE(\sqrt{m \wedge \tilde{m}})} \sum_{i=1}^{\tilde{m}-1} \sqrt{i} M(i) \theta \\&
        = \theta 
    \end{align}
\end{proof}

Then using Lemma \ref{lemma:theta}, we get 
\begin{align}
    &\EE\left(\hat{\theta} \mid \hat{j} \right)-\theta\\
    &= \frac{\sqrt{\tilde{m}}}{ \mathbb{E}_{ }(\sqrt{m \wedge \tilde{m}})} \EE\left(\Pi_{\hat{j}}\left(\hat{X}^{(n)}_{\hat j} \right) \mid \hat j \right) - \frac{ \mathbb{E}_{ }((\sqrt{\tilde{m}}-\sqrt{m})\mathbb{1}(m\leq \tilde{m}))}{\mathbb{E}_{ }(\sqrt{m \wedge \tilde{m}})}s_{\hat j} \\
    &- \left( \frac{\sqrt{\tilde{m}}}{\EE(\sqrt{m \wedge \tilde{m}})}  \EE(\overline{X}^{(n)}_{m_n} \mathbb{1}(m_n\geq \tilde{m})) + \frac{1}{\EE_{ }(\sqrt{m \wedge \tilde{m}})}\sum_{i=1}^{\tilde{m}-1}\EE_{ }(\overline{X}^{(n)}_{m_n}\sqrt{i}\mathbb{1}(m_n =i )) \right).
\end{align}
and using the independence between $(\overline{X}^{(n)}_{m_n}, m_n)$ and $\hat j$, we obtain
\begin{align}
    \EE&\left(\hat{\theta} \mid \hat{j} \right)-\theta=
    \frac{\sqrt{\tilde{m}}}{\mathbb{E}_{m \sim M}(\sqrt{m \wedge \tilde{m}})}
    \EE\left( \left(\Pi_{\hat{j}}\left(\hat{X}^{(n)}_{\hat j} \right) - \overline{X}^{(n)}_{m_n}\right)\mathbb{1}(m_n \geq \tilde{m}) \mid \hat j \right) \notag \\
    &+ \frac{\sqrt{\tilde{m}}}{\mathbb{E}_{m \sim M}(\sqrt{m \wedge \tilde{m}})}  \sum_{i=1}^{\tilde{m}-1}
    \EE\left(\left(\Pi_{\hat{j}}\left(\hat{X}^{(n)}_{\hat j} \right) -  \frac{\sqrt{i}}{\sqrt{\tilde{m}}}\overline{X}^{(n)}_{m_n}\right) \mathbb{1}(m_n = i) \mid \hat j \right) \notag \\
    &- \frac{ \mathbb{E}_{m \sim M}((\sqrt{\tilde{m}}-\sqrt{m})\mathbb{1}(m\leq \tilde{m}))}{\mathbb{E}_{m \sim M}(\sqrt{m \wedge \tilde{m}})}s_{\hat{j}}. \label{eq:bias}
\end{align}

We now introduce two events that will help us further upper bound the bias. Given $l\in \NN^*$ the index of the bin $I_l$ such that $\theta \in I_l$, we define the event that $\hat{j}$ gives the correct bin up to an error of $2$:
\begin{equation}
    \mathcal{A}=\{|\hat{j} - l| \leq 2 \}.
\end{equation}
Note that under $\Acal$, since $s_{\hat j} \in I_{\hat j}$, it follows that $|s_{\hat j} - \theta| \leq 5 \tau$ since $\forall j \in \frac{1}{\tau}, |I_j| = 2 \tau$.

The bias can be decomposed depending on whether the event $\Acal$ holds:
\begin{align}
    \EE[(\EE[\hat{\theta} \mid \hat{j} ]-\theta )^2] &= \EE[(\EE[\hat{\theta} \mid \hat{j} ]-\theta)^2 \ind{A}] + \EE[(\EE [\hat{\theta} \mid \hat{j}]-\theta)^2 \ind{\overline{A}}]  \\
    &\leq \EE[(\EE[\hat{\theta} \mid \hat{j} ]-\theta)^2 \ind{A}] + 4 \PP(\overline{A}) \\
    &= \EE[((\EE[\hat{\theta} \mid \hat{j} ]-\theta)\ind{A})^2] + 4 \PP(\overline{A})
\end{align}

We first focus on upper bounding $\PP(\overline{A})$.

\begin{lemma}[Hoeffding bound \cite{409cf137-dbb5-3eb1-8cfe-0743c3dc925f}]
\label{lemma:Hoeffding-bound}
The random variables $(X_t^{(u)})_{t=1}^i$ have value in $[-1, 1]$ and given $m_u=i$ are independent and identically distributed with mean $\theta$. It holds that
\begin{equation}
        \mathbb{P}\left\{\left|\overline{X}^{(u)}_i - \theta \right| \geq \tau\mid m_u =i\right\} \leq 2 \exp(-i\tau^2/2)
\end{equation}
\end{lemma}

\begin{lemma}[Upper bound on $\PP(\overline{A})$]  It holds that 
\label{lemma:p(abar)}
\begin{equation}
\mathbb{P}(\overline{\mathcal{A}})\leq \frac{1}{\tau}\exp\left( -n\alpha^2\mathbb{P}\left( m_u \geq \tilde{m}\right)^2/579\right).
\end{equation}
\end{lemma}

\begin{proof}

Call $l$ the optimal index i.e. the index such that $\theta \in I_l$.

User $u$ sets $V_l^{(u)} = 0$ because either it does not have enough samples ($m_u \leq \tilde{m}$) or because it has enough samples but its mean estimate is far from $\theta$. Formally,
\begin{equation}
    \{V_l^{(u)}=0 \} \subseteq \left\{m_u \geq \tilde{m}\right\} \cap \left\{ \left| \overline{X}^{(u)} - \theta \right| \geq \tau \right\} \cup \left\{ m_u \leq \tilde{m}\right\}.
\end{equation}

It follows that,
\begin{align}
    \mathbb{P}\left(V_l^{(u)}=0\right) &\leq \sum_{i=\tilde{m}}^{+\infty}M(i)  \mathbb{P}\left\{\left|\overline{X}^{(u)} - \theta \right| \geq \tau\mid m_u =i\right\} + 1-\mathbb{P}(m \geq \tilde{m} ) 
    \\& \leq \sum_{i=\tilde{m}}^{+\infty}M(i) 2\exp(-i\tau^2/2) + 1-\mathbb{P}(m \geq \tilde{m} ) ~~\text{(By \Cref{lemma:Hoeffding-bound})}
    \\& \leq 2 \mathbb{P}\left( m_u \geq \tilde{m}\right) \exp(-\tilde{m}\tau^2/2) +1  - \mathbb{P}\left( m_u \geq \tilde{m}\right).
\end{align}

Let $k \in [2/\tau] \setminus \{l-2, l-1, l, l + 1, l+2 \}$ a sub-optimal index. In order to have $V_k^{(u)} = 1$ for a suboptimal index, the user $u$ must have enough samples but have a mean estimate far from $\theta$.
Formally,
\begin{equation}
    \{V_k^{(u)}=1\} \subseteq \left\{m_u \geq \tilde{m}\right\} \cap \left\{ \left| \overline{X}^{(u)}- \theta \right| \geq \tau \right\},
\end{equation}
and therefore,

\begin{align}
    \mathbb{P}\left(V_k^{(i)} =1\right)& \leq   2 \exp(-\tilde{m}\tau^2/2)\mathbb{P}(m \geq \tilde{m}).
\end{align}

Then, the following equation relates for any $j \in[\frac{1}{\tau}]$, $V^{(u)}_j$ and its privatized version $\widetilde{V}^{(u)}_j$. Writing for any $x>0$, $\pi_x=e^x /\left(1+e^x\right)$, we have that for any $j \in[\frac{1}{\tau}]$,
\begin{align}
    p_j& \defeq \mathbb{P}\left(\widetilde{V}_j^{(i)}=1\right) \\&
    = \pi_{\alpha/6}\mathbb{P}\left(V_j^{(i)}=1\right) + (1-\pi_{\alpha/6})\mathbb{P}\left(V_j^{(i)}=0\right) \\& 
    = \pi_{\alpha/6}\mathbb{P}\left(V_j^{(i)}=1\right) + (1 - \pi_{\alpha/6})(1- \mathbb{P}\left(V_j^{(i)}=1\right))\\&
    =\left(2 \pi_{\alpha / 6}-1\right) \mathbb{P}\left(V_j^{(i)}=1\right)+\left(1-\pi_{\alpha / 6}\right) \notag.
\end{align}

Hence, 

\begin{align}
    p_l & = \mathbb{P}\left(\widetilde{V}_l^{(i)}=1\right) \\&
    = \left(2 \pi_{\alpha / 6}-1\right)\left(1-\mathbb{P}\left(V_l^{(i)}=0\right) \right)  +\left(1-\pi_{\alpha / 6}\right) 
    \\& \geq (2\pi_{\alpha/6} - 1)( \mathbb{P}\left( m_u \geq \tilde{m}\right) -2 \mathbb{P}\left( m_u \geq \tilde{m}\right) \exp(-\tilde{m}\tau^2/2) )  + (1-\pi_{\alpha / 6}) \notag,
\end{align}

and similarly,

\begin{equation}
    p_k \leq 2 (2\pi_{\alpha/6}-1)  \exp(-\tilde{m}\tau^2/2)\mathbb{P}(m \geq \tilde{m})  + \left(1-\pi_{\alpha / 6}\right). \notag
\end{equation}

We therefore get

\begin{align}
    p_l - p_k &\geq  -4 \left(2\pi_{\alpha/6}-1\right)\mathbb{P}\left( m_u \geq \tilde{m}\right)  \exp(-\tilde{m}\tau^2/2)  +\mathbb{P}\left( m_u \geq \tilde{m}\right) (2 \pi_{\alpha / 6}-1)  \notag \\
    &= \left(2\pi_{\alpha/6}-1\right)\mathbb{P}\left( m_u \geq \tilde{m}\right)  \left(1- 4 \exp(-\tilde{m}\tau^2/2)\right)
\end{align}

It holds that $(1- 4 \exp(-\tilde{m}\tau^2/2))\geq 1/2$ , since by definition of $\tau$ in \Cref{eq:tau}, $\tau \geq \sqrt{\frac{2 \ln(8)}{\tilde{m}}}$.

Hence, we get : 
\begin{equation}
    p_l - p_k \geq \left(\pi_{\alpha/6}-1/2\right)\mathbb{P}\left( m_u \geq \tilde{m}\right) \label{eq:lower-pi}
\end{equation}

Letting $Z_j=\sum_{i=1}^{n / 2} \widetilde{V}_j^{(i)}$ denote the total privatised votes for the $j$-th sub-interval.

It holds that: 
\begin{align}
\mathcal{A}&=\{|\hat{j} - l| \leq 2 \} \\
&=\bigcap_{k \in\lceil{\frac{1}{\tau}\rceil}:|l-k|> 2}\left\{Z_l>Z_k\right\} 
\end{align}

We note that
\begin{equation}
Z_l-Z_k=\sum_{i=1}^{n / 2}\left(\widetilde{V}_l^{(u)}-\widetilde{V}_k^{(u)}\right) \notag 
\end{equation}
  where for each $u \in[n / 2]$, it holds that $\widetilde{V}_l^{(u)}-\widetilde{V}_k^{(u)} \in\{-1,0,1\}$ and that $\mathbb{E}\left[\widetilde{V}_l^{(u)}-\widetilde{V}_k^{(u)}\right]=p_l-p_k$. Then
  
\begin{align}
\mathbb{P}\left(Z_l \leq Z_k\right) & =\mathbb{P}\left(\sum_{u=1}^{n / 2}\left(\widetilde{V}_l^{(u)}-\widetilde{V}_k^{(u)}\right)/n \leq 0\right) \notag \\
& \leq \exp(-n/4(p_l-p_k)^2) && \text{(By Hoeffding inequality)} \\
& \leq\exp\left(-n/4 \left(\pi_{\alpha/6}-1/2\right)^2\mathbb{P}\left( m_u \geq \tilde{m}\right)^2\right) && \text{ (By~\Cref{eq:lower-pi}) }
\end{align}

\Cref{lemma:K} that we prove next shows that $\left(\pi_{\alpha/6}-1/2\right)^2 \geq \alpha^2 / 579$ so that we get

\[
\mathbb{P}\left(Z_l \leq Z_j\right) \leq \exp\left( -n\alpha^2 \mathbb{P}\left( m_u \geq \tilde{m}\right)^2/579\right).
\]

\begin{lemma}
\label{lemma:K}
For $\alpha \in [0, 1]$, it holds that
    \begin{equation}
       (1/2 - \pi_{\alpha/6})^2 \geq \alpha^2/  579
    \end{equation}
\end{lemma}

\begin{proof}
Consider the function
    \begin{equation}
    f: x \in [0, 1] \mapsto \frac{x^2}{\left(1/2 -\frac{e^{x/6}}{1+e^{x/6}} \right)^2}
    \end{equation}

Since,

\begin{align}
f(x) &= \frac{x^2}{(1/2 - \frac{\exp(x/6)}{1 + \exp(x/6)})^2} \\
&\sim_{x \rightarrow 0} \frac{x^2}{(1/2 - \frac{1 + x/6}{2 + x/6})^2} \\
&\sim_{x \rightarrow 0} \frac{ x^2}{(1/2 - \frac{1 + x/6}{2}(1 - x/12))^2} \\
&\sim_{x \rightarrow 0} \frac{ x^2}{(1/2 - (1/2 + x /24))^2} \\
&\sim_{x \rightarrow 0}   24^2.
\end{align}
And, $f(x) > 0 \ \forall x \in (0,1]$
We can define :
\begin{equation}
    g: x \in [0,1] \mapsto \ln(f(x))
\end{equation}
$\forall x \in [0,1] $
\begin{align}
    g^{\prime}(x)& =\frac{2}{x}+\frac{1}{\frac{1}{2}-\frac{\exp (x / 6)}{1+\exp (x / 6)}} \cdot \frac{\frac{1}{6} \exp (x / 6)}{(1+\exp (x / 6))^2} \\& = \frac{2}{x} + \frac{2(1+ \exp(x/6)}{(1-\exp(x/6))}\cdot \frac{\frac{1}{6} \exp (x / 6)}{(1+\exp (x / 6))^2}\\&= \frac{2}{x} + \frac{2}{(1-\exp(x/6))}\cdot \frac{\frac{1}{6} \exp (x / 6)}{(1+\exp (x / 6))}  \\& = \frac{2}{x}+ \frac{1/3 \exp(x/6)}{1-\exp(x/3)}  \\& \geq  \frac{2}{x}- \frac{1/3 \exp(1/6)}{\exp(x/3)-1} && \exp(x/6) \leq \exp(1/6) \ \forall x \in [0,1]\\& \geq \frac{2}{x} - \frac{1/3 \exp(1/6)}{x/3} && \exp(x/3)-1 \geq \frac{x}{3} \  \forall x \in \mathbb{R} \\& = \frac{2-\exp(1/6)}{x} \\& \geq 0
\end{align}

Hence, $g$ is non decreasing on $[0,1]$, and so is $f$. We get
\begin{equation}
    \forall x \in [0,1], f(x) \leq f(1) \leq 579.
\end{equation}

\end{proof}

We can then finish the proof of \Cref{lemma:p(abar)} by writing
\begin{align}
    \mathbb{P}\left(\overline{\Acal}\right) &\leq \sum_{j \in \lceil 1/\tau \rceil :|l-j|>2} \mathbb{P}\left(Z_l \leq Z_j\right) \notag && \text{(By a union bound)} \\
    &\leq \frac{1}{\tau}\exp\left( -n\alpha^2\mathbb{P}\left( m_u \geq \tilde{m}\right)^2/579\right) 
\end{align}
\end{proof}

We then focus on $|(\EE[\hat{\theta} \mid \hat{j} ]-\theta)|\ind{A}$. From the triangle inequality, the following equation holds.
\begin{align}
    &|\EE[\hat{\theta} \mid \hat{j} ]-\theta| \ind{A} \\
    &\leq \frac{\sqrt{\tilde{m}}}{\mathbb{E}_{m \sim M}(\sqrt{m \wedge \tilde{m}})}
    \underbrace{\EE\left( \left|\Pi_{\hat{j}}\left(\hat{X}^{(n)}_{\hat j} \right) - \overline{X}^{(n)}_{m_n}\right|\mathbb{1}(m_n \geq \tilde{m})\ind{A} \mid \hat j \right)}_{(I)} \\& + \bigg|\frac{\sqrt{\tilde{m}}}{\mathbb{E}_{m \sim M}(\sqrt{m \wedge \tilde{m}})}  
    \EE\left(\sum_{i=1}^{\tilde{m}-1} \left(\Pi_{\hat{j}}\left(\hat{X}^{(n)}_{\hat j} \right) -  \frac{\sqrt{i}}{\sqrt{\tilde{m}}}\overline{X}^{(n)}_{m_n}\right) \mathbb{1}(m_n = i) \ind{A} \mid \hat j \right) \\& - \frac{ \mathbb{E}_{m \sim M}((\sqrt{\tilde{m}}-\sqrt{m})\mathbb{1}(m\leq \tilde{m}))}{\mathbb{E}_{m \sim M}(\sqrt{m \wedge \tilde{m}})}s_{\hat{j}} \ind{A} \bigg|. && \text{(By \Cref{eq:bias})}
\end{align}

We therefore seek to upper bound the term (I).
\begin{lemma}[Majoration of (I)]
\label{lemma:(I)}
\begin{equation}
    \EE\left( \left|\Pi_{\hat{j}}\left(\hat{X}^{(n)}_{\hat j} \right) - \overline{X}^{(n)}_{m_n}\right|\mathbb{1}(m_n \geq \tilde{m})\ind{A} \mid \hat j \right) \leq 4  \exp(-\tilde{m} \tau^2 / 2) \leq \frac1{4 \sqrt{\tilde{m} n \alpha^2}}
\end{equation}
\end{lemma}

\begin{proof} The proof follows by a similar argument as in \cite{duchi2018minimax} .

We have
\begin{align}
&\EE\left( \left|\Pi_{\hat{j}}\left(\hat{X}^{(n)}_{\hat j} \right) - \overline{X}^{(n)}_{m_n}\right|\mathbb{1}(m_n \geq \tilde{m})\ind{A} \mid \hat j \right) \\
 &=\mathbb{E}\left[\left|\overline{X}^{(n)}_{m_n}-\Pi_{\hat{j}}\left(\overline{X}^{(n)}_{m_n}\right)\right| \mathbb{1}\{m_n \geq \tilde{m}\}\mathbb{1}\{\mathcal{A}\} \mid \hat{j} \right] \\
& =\mathbb{E}\left[\left|\overline{X}^{(n)}_{m_n}-\tilde{U}_{\hat j}\right| \mathbb{1}\left\{\overline{X}^{(n)}_{m_n}>\tilde{U}_{\hat j}\right\} \mathbb{1}\{m_n \geq \tilde{m}\}\mathbb{1}\{\mathcal{A}\} \mid \hat{j}\right] \notag \\
&+\mathbb{E}\left[\left|\overline{X}^{(n)}_{m_n}-\tilde{L}_{\hat j}\right| \mathbb{1}\left\{\overline{X}^{(n)}_{m_n}<\tilde{L}_{\hat j}\right\}\mathbb{1}\{m_n \geq \tilde{m}\} \mathbb{1}\{\mathcal{A}\} \mid \hat{j} \right] \\
& \leq 2 \mathbb{P}\left(\left\{\overline{X}^{(n)}_{m_n}>\tilde{U}_{\hat j}\right\} \cap \{m_n \geq \tilde{m}\}\cap \mathcal{A} \mid \hat{j} \right) \notag \\
&+2 \mathbb{P}\left(\left\{\overline{X}^{(n)}_{m_n}<\tilde{L}_{\hat j}\right\} \cap \{m_n \geq \tilde{m}\} \cap \mathcal{A} \mid \hat{j} \right),
\end{align}

where in the last inequality uses that $|\overline{X}^{(n)}_{m_n} - \tilde{L}_{\hat j}| \leq 2$ and $|\overline{X}^{(n)}_{m_n} - \tilde{U}_{\hat j}| \leq 2$ which follows since $\overline{X}^{(n)}_{m_n}, \tilde{L}_{\hat j}, \tilde{U}_{\hat j} \in [-1, 1]$. Then, we upper bound $\mathbb{P}\left(\left\{\overline{X}^{(n)}_{m_n}>\tilde{U}_{\hat j}\right\} \cap \{m_n \geq \tilde{m}\} \cap \mathcal{A} \mid \hat{j} \right)$.
\begin{align}
&\mathbb{P}\left(\left\{\overline{X}^{(n)}_{m_n}>\tilde{U}_{\hat j}\right\} \cap \{m_n \geq \tilde{m}\} \cap \mathcal{A} \mid \hat{j} \right) 
\\& =\mathbb{P}\left(\bigcup_{i=\tilde{m}}^{+\infty} \left\{\overline{X}^{(n)}_{i}>\tilde{U}_{\hat j}\right\} \cap \{m_n = i\} \cap \mathcal{A} \mid \hat{j} \right)
\\&  = \sum_{i=\tilde{m}}^{+ \infty} \mathbb{P}\left( \left\{\overline{X}^{(n)}_{i}>\tilde{U}_{\hat{j}}\right\} \mid \{m_n = i\},  \hat j \right) M(i) \mathbb{P}\left( \mathcal{A}\right)
\\&\leq \sum_{i=\tilde{m}}^{+ \infty} \mathbb{P}\left( \left\{\overline{X}^{(n)}_{i}>\tilde{U}_{\hat{j}}\right\} \mid \{m_n = i\},  \hat j \right) M(i) \\
&=  \sum_{i=\tilde{m}}^{+ \infty} \mathbb{P}\left( \left\{\overline{X}^{(n)}_{i}>\tilde{U}_{\hat{j}}\right\} \mid \{m_n = i\} \right) M(i) && \text{(Since $\hat j$ independent of $\overline{X}_i^{(n)}$ given $m_n=i$)} \\
&\leq  \sum_{i=\tilde{m}}^{+ \infty} \mathbb{P}\left( \left\{\overline{X}^{(n)}_{i}> \theta + \tau \right\} \mid \{m_n = i\} \right) M(i)  \\
&\leq \sum_{i=\tilde{m}}^{+ \infty} \exp(-i (\tau)^2/2) M(i) && \text{(Hoeffding inequality)} \\
&\leq \exp(- \tilde{m} \tau^2 / 2)
\end{align}

By a symmetric reasoning, we show that 
\[
\mathbb{P}\left(\left\{\overline{X}^{(n)}_{m_n}<\tilde{L}_{\hat j}\right\} \cap \{m_n \geq \tilde{m}\} \cap \mathcal{A} \mid \hat{j} \right) \leq \exp(- \tilde{m} \tau^2 / 2).
\]

Together the two bounds yield the result.
\end{proof}

We now introduce, for any $a \in [\tilde{m} - 1]$ the events $B_a$ that describe how $\overline{X}^{(n)}_a$ concentrates around $\theta$:
\begin{equation}
    \forall a \in [\tilde{m}], \mathcal{B}_a = \left\{ \left| \overline{X}^{(n)}_a - \theta\right| \leq \frac{\sqrt{2\ln(8 (\sqrt{\tilde{m} n \alpha^2} \vee 1))}}{\sqrt{a}} \right\}.
\end{equation}

By using \Cref{lemma:(I)} and incorporating the events $(B_a)_{a \in [\tilde{m} - 1]}$, the upper bound on $|\EE[\hat{\theta} \mid \hat{j} ]-\theta| \ind{\Acal}$ becomes:
\begin{align}
    &|\EE[\hat{\theta} \mid \hat{j} ]-\theta| \ind{\mathcal{A}} \\
    &\leq \frac{\sqrt{\tilde{m}}}{\mathbb{E}_{m \sim M}(\sqrt{m \wedge \tilde{m}})}   4  \exp(-\tilde{m} \tau^2 / 2) \\
    &+ \bigg|\frac{\sqrt{\tilde{m}}}{\mathbb{E}_{m \sim M}(\sqrt{m \wedge \tilde{m}})}  
    \EE\left(\left[\sum_{i=1}^{\tilde{m}-1} \left(\Pi_{\hat{j}}\left(\hat{X}^{(n)}_{\hat j} \right) -  \frac{\sqrt{i}}{\sqrt{\tilde{m}}}\overline{X}^{(n)}_{m_n}\right) \mathbb{1}(m_n = i) - \frac{ \mathbb{E}_{m \sim M}((\sqrt{\tilde{m}}-\sqrt{m})\mathbb{1}(m\leq \tilde{m}))}{\sqrt{\tilde{m}}}s_{\hat{j}} \right] \ind{\mathcal{A}} \mid \hat j \right)  \bigg| \\
    &= \frac{\sqrt{\tilde{m}}}{\mathbb{E}_{m \sim M}(\sqrt{m \wedge \tilde{m}})}   4  \exp(-m \tau^2 / 2) \\
    &+ \frac{\sqrt{\tilde{m}}}{\mathbb{E}_{m \sim M}(\sqrt{m \wedge \tilde{m}})}  
    \underbrace{\bigg|\EE\left(\sum_{i=1}^{\tilde{m}-1} \left(\Pi_{\hat{j}}\left(\hat{X}^{(n)}_{\hat j} \right)  -  \frac{\sqrt{i}}{\sqrt{\tilde{m}}}\overline{X}^{(n)}_{m_n} - \frac{ (\sqrt{\tilde{m}}-\sqrt{i})}{\sqrt{\tilde{m}}}s_{\hat{j}}\right) \ind{\mathcal{A}, \mathcal{B}_i} \mathbb{1}(m_n = i) \mid \hat j \right)  \bigg|}_{(II)} \\
    &+ \frac{\sqrt{\tilde{m}}}{\mathbb{E}_{m \sim M}(\sqrt{m \wedge \tilde{m}})}  
    \underbrace{\bigg|\EE\left(\sum_{i=1}^{\tilde{m}-1} \left(\Pi_{\hat{j}}\left(\hat{X}^{(n)}_{\hat j} \right)  -  \frac{\sqrt{i}}{\sqrt{\tilde{m}}}\overline{X}^{(n)}_{m_n} - \frac{ (\sqrt{\tilde{m}}-\sqrt{i})}{\sqrt{\tilde{m}}}s_{\hat{j}} \right) \ind{\mathcal{A}, \overline{\mathcal{B}_i}} \mathbb{1}(m_n = i) \mid \hat j \right)  \bigg|}_{(III)}
\end{align}

Let us start by upper bounding $(III)$.
\begin{lemma}[Majoration of $(III)$]
    \[
    (III) \leq \frac1{2 \sqrt{\tilde{m} n \alpha^2}}
    \]
    \label{lemma:iii}
\end{lemma}
\begin{proof}
Since $\overline{X}^{(n)}_{m_n} \in [-1, 1]$, $\Pi_{\hat{j}}\left(\hat{X}^{(n)}_{\hat j} \right) \in [-1, 1]$  and $s_{\hat{j}} \in [-1, 1]$, we have that 
\[
\Pi_{\hat{j}}\left(\hat{X}^{(n)}_{\hat j} \right)  -  \frac{\sqrt{i}}{\sqrt{\tilde{m}}}\overline{X}^{(n)}_{m_n} - \frac{ (\sqrt{\tilde{m}}-\sqrt{i})}{\sqrt{\tilde{m}}}s_{\hat{j}} \in [-2, 2]
\]
and therefore 
\begin{align}
    (III) &= \bigg|\EE\left(\sum_{i=1}^{\tilde{m}-1} \left(\Pi_{\hat{j}}\left(\hat{X}^{(n)}_{\hat j} \right)  -  \frac{\sqrt{i}}{\sqrt{\tilde{m}}}\overline{X}^{(n)}_{m_n} - \frac{ (\sqrt{\tilde{m}}-\sqrt{i})}{\sqrt{\tilde{m}}}s_{\hat{j}} \right) \ind{A, \overline{B_i}} \mathbb{1}(m_n = i) \mid \hat j \right)  \bigg| \\
    &\leq 2 \sum_{i=1}^{\tilde{m}} \PP(\overline{B_i}, m_n=i \mid \hat j) \\
    &= 2 \sum_{i=1}^{\tilde{m}} \PP(\overline{B_i}, m_n=i) ~~~ \text{(Since $\hat j$ and $(m_n, B_i)$ are independent)} \\
    &= 2 \sum_{i=1}^{\tilde{m}} \PP(\overline{B_i} | m_n=i) M(i) \\
    &= 2 \sum_{i=1}^{\tilde{m}} \PP(|\overline{X}^{(n)}_i - \theta| >  \frac{\sqrt{2\ln(8 (\sqrt{\tilde{m} n \alpha^2} \vee 1))}}{\sqrt{a}} | m_n=i) M(i) \\
    &\leq \frac1{2 \sqrt{\tilde{m} n \alpha^2}} \sum_{i=1}^{\tilde{m}} M(i) && \text{(By Hoeffding bound)} \\
    &\leq \frac1{2 \sqrt{\tilde{m} n \alpha^2}}
\end{align}
\end{proof}

Then, we focus on $(II)$. Under $\Acal$ and $\Bcal_a$, \Cref{lemma:concentration-few-samples} upper bounds the distance between $\hat{X}^{(n)}_{\hat j}$ and $s_{\hat j}$.

\begin{lemma}[Concentration of users with few samples]
\label{lemma:concentration-few-samples}
$\forall a \leq \tilde{m}$, it holds that
\begin{equation}
    \mathcal{A} \cap \mathcal{B}_a \implies \left| \hat{X}^{(n)}_{\hat j} - s_{\hat{j}} \right|\ind{m_n=a}  \leq 6 \tau.
\end{equation}
\end{lemma}
\begin{proof}
Fix $a \leq \tilde{m}$. Under $\Acal$ and $\Bcal_a$, it holds that
    \begin{align}
        \left| \hat{X}^{(m)}_{\hat j} - s_{\hat{j}} \right|\ind{m_n=a} 
        & =  \left| \frac{\sqrt{a}}{\sqrt{\tilde{m}}}\left( \overline{X}^{(n)}_{a}+ \left(\frac{\sqrt{\tilde{m}}}{\sqrt{a}} -1\right) s_{\hat{j}}  \right) - s_{\hat{j}} \right|\ind{m_n=a}\\
        &\leq \left| \frac{\sqrt{a}}{\sqrt{\tilde{m}}}\left( \overline{X}^{(n)}_{a}+ \left(\frac{\sqrt{\tilde{m}}}{\sqrt{a}} -1\right) s_{\hat{j}}  \right) - s_{\hat{j}} \right|\\
        &  =   \left| \frac{\sqrt{a}}{\sqrt{\tilde{m}}}\left( \overline{X}^{(n)}_{a} - s_{\hat{j}} \right) \right|
        \\& \leq \frac{\sqrt{a}}{\sqrt{\tilde{m}}}\left( \left| \overline{X}_a^{(n)} - \theta \right| + \left| \theta - s_{\hat{j}} \right| \right)
        \\& \leq \frac{\sqrt{a}}{\sqrt{\tilde{m}}}\left( \frac{\sqrt{2\ln(8 (\sqrt{\tilde{m} n \alpha^2} \vee 1))}}{\sqrt{a}} +  5 \tau \right) && \text{ Valid under } \mathcal{B}_a \cap \mathcal{A}
        \\& \leq 6 \tau
    \end{align}
\end{proof}

From \Cref{lemma:concentration-few-samples} we know that $\hat{X}^{(n)}_{\hat j}$ is at distance at most of $6\tau$ of $s_{\hat{j}}$, the center of $I_{\hat{j}}$. Furthermore, $\hat{X}^{(m)}_{\hat j} \in [-1, 1]$ Hence, as $\left[\tilde{L}_{\hat{j}},\tilde{U}_{\hat{j}}\right] = \left[ -1 \vee (s_{\hat{j}}- 6 \tau), 1 \wedge (s_{\hat{j}}+ 6\tau) \right] $ we deduce that $\hat{X}^{(n)}_{\hat j} \in \left[\tilde{L}_{\hat{j}},\tilde{U}_{\hat{j}}\right] $. From the definition of $\Pi_{\hat j}$ which is the projection operator on $\left[\tilde{L}_{\hat{j}}, \tilde{U}_{\hat{j}}\right]$, it holds that 
\begin{align}
\forall a < \tilde{m}, \Pi_{\hat j}(\hat{X}^{(n)}_{\hat j}) \ind{m_n = a, \Acal, \Bcal_a} &= \hat{X}^{(n)}_{\hat j}\ind{m_n=a, \Acal, \Bcal_a} \notag \\
&= \frac{\sqrt{a}}{\sqrt{\tilde{m}}}\left( \overline{X}^{(n)}_{a} + \left(\frac{\sqrt{\tilde{m}}}{\sqrt{a}} -1\right) s_{\hat{j}}  \right)\ind{m_n=a, \Acal, \Bcal_a} \\
&= \left( \frac{\sqrt{a}}{\sqrt{\tilde{m}}} \overline{X}^{(n)}_{a} + \left(1 -\frac{\sqrt{a}}{\sqrt{\tilde{m}}}\right) s_{\hat{j}}  \right)\ind{m_n=a, \Acal, \Bcal_a} \label{eq:proj under Ba}
\end{align}

We are now ready to prove that $(II) = 0$.
\begin{lemma}[Nullity of $(II)$]
    \[
    (II) = 0
    \]
    \label{lemma:ii}
\end{lemma}
\begin{proof}
    \begin{align}
        &(II) \\
        &= \bigg|\EE\left(\sum_{i=1}^{\tilde{m}-1} \left(\Pi_{\hat{j}}\left(\hat{X}^{(n)}_{\hat j} \right)  -  \frac{\sqrt{i}}{\sqrt{\tilde{m}}}\overline{X}^{(n)}_{i} - \frac{ (\sqrt{\tilde{m}}-\sqrt{i})}{\sqrt{\tilde{m}}}s_{\hat{j}}\right) \ind{\mathcal{A}, \mathcal{B}_i} \mathbb{1}(m_n = i) \mid \hat j \right)  \bigg| \\
        &= \bigg|\EE \left(\sum_{i=1}^{\tilde{m}-1} \left( \left(\frac{\sqrt{i}}{\sqrt{\tilde{m}}} \overline{X}^{(n)}_{i} + (1 -\frac{\sqrt{i}}{\sqrt{\tilde{m}}}) s_{\hat{j}} \right)  -  \frac{\sqrt{i}}{\sqrt{\tilde{m}}}\overline{X}^{(n)}_{i} - \frac{ (\sqrt{\tilde{m}}-\sqrt{i})}{\sqrt{\tilde{m}}}s_{\hat{j}} \right)\ind{A, B_i} \mathbb{1}(m_n = i) \mid \hat j \right)  \bigg| && \text{(By \cref{eq:proj under Ba})} \\&
        = \bigg|\EE \left(\sum_{i=1}^{\tilde{m}-1} \left( \frac{\sqrt{i}}{\sqrt{\tilde{m}}} \overline{X}^{(n)}_{i} -  \frac{\sqrt{i}}{\sqrt{\tilde{m}}}\overline{X}^{(n)}_{i} + \left(1 -\frac{\sqrt{i}}{\sqrt{\tilde{m}}}\right) s_{\hat{j}}   -\left( 1- \frac{ (\sqrt{i})}{\sqrt{\tilde{m}}}\right)s_{\hat{j}} \right)\ind{A, B_i} \mathbb{1}(m_n = i) \mid \hat j \right)  \bigg|\\
        &= 0
    \end{align}
\end{proof}

From \Cref{lemma:iii} and \Cref{lemma:ii}, we obtain 
\begin{align}
|\EE[\hat{\theta} \mid \hat{j} ]-\theta| \ind{\mathcal{A}} &= \frac{\sqrt{\tilde{m}}}{\mathbb{E}_{m \sim M}(\sqrt{m \wedge \tilde{m}})} \frac1{\sqrt{\tilde{m} n \alpha^2}} \notag \\
    &= \frac1{\sqrt{\mathbb{E}_{m \sim M}(\sqrt{m \wedge \tilde{m}}) n \alpha^2}} \label{eq:risk under A}
\end{align}

Incorporating the result from \Cref{lemma:p(abar)}, we get
\begin{align}
    \EE[(\EE[\hat{\theta} \mid \hat{j} ]-\theta )^2] &\leq  \frac1{\mathbb{E}_{m \sim M}(\sqrt{m \wedge \tilde{m}}) n \alpha^2} + \frac{4}{\tau}\exp\left( -n\alpha^2\mathbb{P}\left( m_u \geq \tilde{m}\right)^2/579\right)
\end{align}

We then deal with the variance term $\Var\left(\hat{\theta} \mid \hat{j}\right)$.
Since 
\[
\hat{\theta}=  \frac{\sqrt{\tilde{m}} \frac{2}{n} \sum_{u=n/2+1}^n\hat{\theta}^{(u)}}{ \mathbb{E}_{m \sim M}(\sqrt{m \wedge \tilde{m}})} - \frac{ \mathbb{E}_{m \sim M}((\sqrt{\tilde{m}}-\sqrt{m})\mathbb{1}(m\leq \tilde{m}))}{\mathbb{E}_{m \sim M}(\sqrt{m \wedge \tilde{m}})}s_{\hat{j}},
\]
and $(\hat \theta^{(u)})_{u=n/2 + 1}^n$ are independent given $\hat{j}$, it holds that
\[
\Var\left(\hat{\theta} \mid \hat{j}\right) = \frac{ 2 \tilde{m}}{n \mathbb{E}_{m \sim M}(\sqrt{m \wedge \tilde{m}})^2} \Var(\hat \theta^{(n)} \mid \hat j)
\]

\begin{lemma}[Majoration of $\Var(\hat \theta^{(n)} | \hat j)$]
    \begin{equation}
        \operatorname{Var}\left(\hat{\theta}^{(n)} \mid \hat{j} \right) \leq  \frac{2(14 \tau)^2}{\alpha^2}
    \end{equation}
    \label{lemma:(V)}
\end{lemma}

\begin{proof}
\begin{align}
    \operatorname{Var}\left(\hat{\theta}^{(n)} \mid \hat{j} \right)&=\operatorname{Var}\left(\overline{X^{(n)}} + \frac{14 \tau}{\alpha} \ell_n \mid \hat{j}\right) 
    \\& =  \operatorname{Var}\left(\Pi_{\tilde{I}_{\hat{j} }}\left(\overline{X^{(n)}}\right) \right) + \operatorname{Var}\left(\frac{14 \tau}{\alpha} \ell_n \right) 
    \\& \leq  14^2 \tau^2+\frac{ 14^2 \tau^2}{\alpha^2}  
    \\& \leq \frac{14^2\cdot 2 \tau^2}{\alpha^2}
\end{align}
\end{proof}

The upper bound of $\EE[(\hat \theta - \theta)^2]$ is therefore given by:
\begin{align}
    \EE[(\EE[\hat{\theta} \mid \hat{j} ]-\theta )^2] &\leq  \frac1{\mathbb{E}_{m \sim M}(\sqrt{m \wedge \tilde{m}}) n \alpha^2}\\
    &+ \frac{4}{\tau}\exp\left( -n\alpha^2\mathbb{P}\left( m_u \geq \tilde{m}\right)^2/579\right) + \frac{\tilde{m}}{(\mathbb{E}_{m \sim M}(\sqrt{m \wedge \tilde{m}}))^2}\frac{14^2\cdot 4 \tau^2}{n \alpha^2} \\
    &\leq  \frac{(14^2 \times 4 \times 2 + 1) \ln(8 (\sqrt{\tilde{m} n \alpha^2} \vee 1))}{n \alpha^2 (\mathbb{E}_{m \sim M}(\sqrt{m \wedge \tilde{m}}))^2} \\
    &+ 4 \sqrt{\frac{\tilde{m}}{2\ln(8 (\sqrt{\tilde{m} n \alpha^2} \vee 1))}} \exp\left( -n\alpha^2\mathbb{P}\left( m_u \geq \tilde{m}\right)^2/579\right)
\end{align}

Recall that $\tilde{m}= \argmax_{a \in \NN^*}\{\mathbb{P}_{m \sim M}\left( m \geq a \right) \geq \left(\frac{579}{n \alpha^2} \frac{3}{2}\ln(\frac{8 (a n\alpha^2 \vee 1)}{\ln(8 (a n\alpha^2 \vee 1))})\right) \wedge 1\}$.
If $\tilde{m} > 1$, then it holds that
\[
\exp\left( -n\alpha^2\mathbb{P}\left( m_u \geq \tilde{m}\right)^2/579\right) \leq \frac{\ln(8 (\tilde{m} n\alpha^2 \vee 1))^{3/2}}{(8 (\tilde{m} n\alpha^2 \vee 1))^{3/2}} \leq \frac{\ln(8 (\tilde{m} n\alpha^2 \vee 1))^{3/2}}{8 \tilde{m}^{3/2} n \alpha^2},
\]
and therefore
\[
4 \sqrt{\frac{\tilde{m}}{2\ln(8 (\sqrt{\tilde{m} n \alpha^2} \vee 1))}} \exp\left( -n\alpha^2\mathbb{P}\left( m_u \geq \tilde{m}\right)^2/579\right) \leq \frac{\ln(8 (\tilde{m} n\alpha^2 \vee 1))}{\tilde{m} n \alpha^2}.
\]

As a result, we obtain the bound
\begin{align}
     \EE[(\EE[\hat{\theta} \mid \hat{j} ]-\theta )^2] &\leq \frac{(14^2 \times 4 \times 2 + 2) \ln(8 (\sqrt{\tilde{m} n \alpha^2} \vee 1))}{n \alpha^2 (\mathbb{E}_{m \sim M}(\sqrt{m \wedge \tilde{m}}))^2} \notag \\
    &= \frac{1570  \ln(8 (\sqrt{\tilde{m} n \alpha^2} \vee 1))}{n \alpha^2 (\mathbb{E}_{m \sim M}(\sqrt{m \wedge \tilde{m}}))^2}. \label{eq:final risk}
\end{align}

Otherwise, $\tilde{m} = 1$ and therefore $\tau \geq 1$ which has as an immediate consequence that $\PP(\Acal) = 1$ (there is only one bin so $\hat j$ is necessarily the index of the optimal bin). Then, from the upper bound on $|\EE[\hat{\theta} \mid \hat{j} ]-\theta| \ind{\mathcal{A}}$ in \Cref{eq:risk under A}, we obtain that 
\begin{align}
    \EE[(\EE[\hat{\theta} \mid \hat{j} ]-\theta )^2] &\leq  \frac1{\mathbb{E}_{m \sim M}(\sqrt{m \wedge \tilde{m}}) n \alpha^2}.
\end{align}

Incorporating the upper bound on the variance from \Cref{lemma:(V)}, we obtain the bound
\[
 \EE[(\EE[\hat{\theta} \mid \hat{j} ]-\theta )^2] \leq \frac{1569 \ln(8 (\sqrt{\tilde{m} n \alpha^2} \vee 1))}{n \alpha^2 (\mathbb{E}_{m \sim M}(\sqrt{m \wedge \tilde{m}}))^2}
\]
which together with the bound in \Cref{eq:final risk} concludes the proof.
\end{proof}

\section{Assets used}
\label{app:assets}
To conduct our numerical experiments we used Numpy  and Statsmodels for computations \cite{harris2020array} \cite{seabold2010statsmodels} which have a BSD compatible lience and matplotlib for visualisation
\cite{Hunter:2007} which has a PSF compatible licence.

\end{document}